\documentclass[12pt,a4paper]{article}
\usepackage{amsmath,amsthm,amsfonts,amssymb,bbm}
\usepackage{graphicx,psfrag,subfigure,color}
\usepackage{cite}
\usepackage{hyperref}

\DeclareMathOperator{\Cov}{Cov}
\DeclareMathOperator{\Var}{Var}

\newcommand{\iid}{i.i.d.\ }

\newcommand{\LB}{L^{\mathcal B}}
\newcommand{\Or}{\mathcal{O}}

\newcommand{\Pb}{\mathbb{P}}
\newcommand{\E}{\mathbbm{E}}
\newcommand{\Id}{\mathbbm{1}}
\newcommand{\e}{\varepsilon}

\newcommand{\R}{\mathbb{R}}

\newcommand{\Z}{\mathbb{Z}}

\newtheorem{prop}{Proposition}[section]
\newtheorem{thm}[prop]{Theorem}
\newtheorem{lem}[prop]{Lemma}
\newtheorem{defin}[prop]{Definition}
\newtheorem{cor}[prop]{Corollary}

\newtheorem{cla}[prop]{Claim}

\newtheorem{rem}[prop]{Remark}
\newenvironment{remark}{\begin{rem}\normalfont}{\end{rem}}

\numberwithin{equation}{section}

\title{Time-time covariance for last passage\\ percolation with generic initial profile}
\author{P.L. Ferrari\thanks{Institute for Applied Mathematics, Bonn University, Endenicher Allee 60, 53115 Bonn, Germany. E-mail: {\tt ferrari@uni-bonn.de}}
\and  A. Occelli\thanks{Institute for Applied Mathematics, Bonn University, Endenicher Allee 60, 53115 Bonn, Germany. E-mail: {\tt occelli@iam.uni-bonn.de}}}
\date{December 20, 2018}

\begin{document}
\sloppy
\maketitle

\begin{abstract}
We consider time correlation for KPZ growth in 1+1 dimensions in a neighborhood of a characteristics. We prove convergence of the covariance with droplet, flat and stationary initial profile. In particular, this provides a rigorous proof of the exact formula of the covariance for the stationary case obtained in~\cite{FS16}. Furthermore, we prove the universality of the first order correction when the two observation times are close and provide a rigorous bound of the error term. This result holds also for random initial profiles which are not necessarily stationary.
\end{abstract}

\section{Introduction}
Stochastic growth models in the Kardar-Parisi-Zhang (KPZ) universality class~\cite{KPZ86} on a one-dimensional substrate are described by a height function $h(x,t)$ with $x$ denoting space and $t$ time. The height function evolves microscopically according to a random and local dynamics, while on a macroscopic scale the evolution is a deterministic PDE and the limit shape is non-random. In particular, if the speed of growth as a function of the gradient of the interface is a strictly convex or concave function, then the model is in the KPZ universality class. One expects large time universality under an appropriate scaling limit.

By studying special models in the KPZ class, the law of the one-point fluctuations and of the spatial statistics are well-known. In particular, the fluctuations scales as $t^{1/3}$ and the correlation length as $t^{2/3}$ (see surveys and lecture notes~\cite{FS10,Cor11,QS15,BG12,Qua11,Fer10b,Tak16})\footnote{This holds true around point with smooth limit shape. Around shocks there are some differences, see e.g.~\cite{FN13,FN16,FF94b}.}. Furthermore, it is known that non-trivial correlations survive on the macroscopic time scale if one considers space-time points along characteristic lines of the PDE for the macroscopic evolution~\cite{Fer08,CFP10b}. This phenomenon is called slow-decorrelation and it indicates that non-trivial processes in a spatial $t^{2/3}$-neighborhood of a characteristic and for macroscopic temporal scale is to be expected. The limit process depends on the initial condition, since this is already the case for the processes at a fixed time.

The study of the time-time process started much more recently. On the experimental and numerical simulation side observables like the persistence probability or the covariance of an appropriately rescaled height function have been studied~\cite{TS12,Tak13,TA16,Tak12}. On the analytic and rigorous side, the two-time joint distribution of the height function is known for special initial conditions: Johansson analyzed a model on full space~\cite{Jo15,Jo18}, while Baik and Liu considered a model on a torus~\cite{BL16,BL17}. There are also non-rigorous works on the time-time covariance and on the upper tail of distributions using replica approach~\cite{ND17,ND18,NDT17}. For general (random) initial conditions exact formulas on the joint distributions are not yet available. Also, the analysis of the covariance starting from the available formulas~\cite{Jo18,BL17} seems to be a difficult task.

In~\cite{FS16} Ferrari and Spohn made some predictions for the behavior of the two-time covariance for three typical initial conditions based on a last passage percolation (LPP) picture. In particular, for the stationary case, an exact formula for the covariance of two points along a characteristic has been derived. Furthermore, the behavior when the macroscopic times were either close or far from each others were provided. However, the work is not mathematically rigorous since the exchange of the large time limit and maximum over sums of Airy processes as well as justification for convergence of the covariances are not provided. The work by Corwin, Liu and Wang~\cite{CLW16} showed the way to obtain a rigorous convergence of distribution in terms of the variational process used in~\cite{FS16}, by lifting the finite-dimensional slow-decorrelation result of~\cite{Fer08,CFP10b} to a functional slow-decorrelation statement.

In this paper we consider a last passage percolation model, which can be also seen as a (version of the) polynuclear growth model. As initial condition we consider the three standard cases (called droplet, flat and stationary) as in~\cite{FS16}, but we extend the study to random but not stationary initial profiles (see~\cite{CFS16} for a related model). In the first three cases by the method of~\cite{CLW16} (simplified in some aspects in~\cite{CFS16,FO17}) one knows that the limiting distributions of (rescaled) LPP times can be expressed as a variational problem in terms of some Airy processes. The first result proven in this paper is the convergence of the covariance of the LPP time to the covariance of the limiting processes, see Theorem~\ref{Thm1}. As a corollary, this provides a proof for the exact formula of the covariance for the stationary case of~\cite{FS16}. We actually extend the result by taking points not exactly on the characteristics, but in a $t^{2/3}$-neighborhood of it.

Our second result concerns the behavior of the covariance when the two times are close to each other on a macroscopic scale. Physically we expect to see the signature of the stationary state as first approximation. This was noticed also in numerical experiments~\cite{Tak13}. This is proven in Theorem~\ref{thm:covGeneral} for all the initial conditions considered. We also provide a rigorous error term, which is compatible with the experiments\footnote{The next order correction is sensitive to the scaling used to define the process. For the scaling used in this paper the error term seems to be optimal. However, if one scales the random variables to have the same one-point distribution function, then experimentally the error term is smaller: instead of an error term with exponent $1^-$, one gets an exponent $\min\{5/3, 2/3 + \alpha\}$, where $\alpha$ is the exponent controlling the convergence of the variance of the height difference to that of the Baik-Rains distribution~\cite{TakPriv18}.}. To obtain the result, we need to control the spatial process at fixed time on small scales. This is achieved by comparing with stationary cases on sets of high probability. The idea goes back to Cator and Pimentel~\cite{CP15b} for the droplet case (extended to general case in~\cite{Pi17}). The control on the high probability sets requires bounds on exit point probabilities, which has to be obtained for each initial profile. In particular, to achieve a good control in the error term, one can not use soft bounds as in~\cite{BCS06,Pi17}. Finally, for droplet initial condition we derive a result also when times are far apart, see Theorem~\ref{thm:covstep0}.

A few weeks after we finished our paper, for the droplet geometry, Basu and Ganguly obtained the same exponents for the behaviour at close or far away points~\cite{BG18}. Unlike in our paper, they did not identify the prefactor, but on the other hand, their result are non-asymptotic as well. One input often used in their paper are the bounds on transversal fluctuations of~\cite{BSS14}.

\bigskip

\emph{Outline:} In Section~\ref{sectModel} we introduce the model, state some known limiting results necessary for the rest of the paper and provide the main results. In Section~\ref{sec:stationaryComparison} we recall the stationary LPP and the comparison lemmas. In Section~\ref{sec:ConvergenceCovariance} we prove Theorem~\ref{Thm1} on the convergence of the covariance. In Section~\ref{sec:CloseTimeBehavior} we prove Theorem~\ref{thm:covGeneral} on the close time behaviour, while in Section~\ref{sec:SeparateTimeBehavior} we sketch the proof of Theorem~\ref{thm:covstep0}. The appendix contains several bounds on the one-point distribution or on increments, which are used in the proofs.

\bigskip

\emph{Acknowledgments:} This work is supported by the German Research Foundation in the Collaborative Research Center 1060 ”The Mathematics of Emergent Effects”, project B04.

\section{Model and results}\label{sectModel}

\subsection{LPP and polynuclear growth}\label{sec:PNGandLPP}
Consider a collection of \iid random variables ${\omega_{i,j},i,j\in \Z}$ with exponential distribution of parameter one. An \emph{up-right path} $\pi=(\pi(0),\pi(1),\ldots,\pi(n))$ on $\Z^2$ from a point $A$ to a point $E$ is a sequence of points in $\Z^2$ with \mbox{$\pi(k+1)-\pi(k)\in \{(0,1),(1,0)\}$}, with $\pi(0)=A$ and $\pi(n)=E$, and $n$ is called the length $\ell(\pi)$ of $\pi$. Given a set of points $S_A$ with some random variables (not necessarily independent) $h^0$ on $S_A$, but independent of the $\omega$'s, and given a point $E$, one defines the last passage time $L_{S_A\to E}$ as
\begin{equation}\label{eq1i}
 L_{S_A\rightarrow E}=\max_{\begin{smallmatrix} \pi:A\rightarrow E \\ A \in S_A \end{smallmatrix}} \bigg( h^0(\pi(0))+\sum_{1\le k\le n} \omega_{\pi(k)}\bigg).
\end{equation}
Also, for two points $P,Q$ which are not on the initial set $S_A$, we define $L_{P\to Q}$ as above but without the term $h^0(\pi(0))$.
$\pi_{S_A\rightarrow E}^{max}$ indicates the maximizer of the last passage time. For continuous random variables, the maximizer is a.s.\ unique\footnote{The only exception will be if $h^0$ is not random, since then the maximizer is unique up to the initial point, which has weight $0$ and thus it is irrelevant.}.

LPP can be though as a stochastic growth model, a version of the polynuclear growth model, as follows.
Let $S_A={\cal L}:=\{(i,j)\in\Z^2\, | \, i+j=0\}$ and let $h^0$ represents a height function at time $t=0$. Then one defines the height function at time $t$ by the relation
\begin{equation}
h(x,t)=L_{{\cal L}\to ((x+t)/2,(t-x)/2)}
\end{equation}
for all $x-t$ being even numbers (and set $h(x,t)=L_{{\cal L}\to ((x+t-1)/2,(t-x-1)/2)}$ for $x-t$ odd). The dynamics of the height function is
\begin{equation}
h(x,t)=\max\{h(x-1,t-1),h(x,t-1),h(x+1,t-1)\}+\omega_{(x+t)/2,(t-x)/2}
\end{equation}
with initial conditions $h(x,0)=h^0(x/2,-x/2)$ (here $\omega_{(x+t)/2,(t-x)/2}=0$ if $x-t$ is odd).

We are interested in the scaling limit of the height function
\begin{equation}
 (w,\tau)\mapsto \lim_{t\to\infty}\frac{h(w 2^{1/3} t^{2/3},\tau t)-\tau t}{2^{2/3}t^{1/3}}
\end{equation}
or, equivalently, setting $E=(\tau N,\tau N)+w(2N)^{2/3}(1,-1)$,
\begin{equation}
 (w,\tau)\mapsto \lim_{N\to\infty}\frac{L_{S_A\to E}-4\tau N}{2^{4/3}N^{1/3}},
\end{equation}
for different initial conditions\footnote{The choice of zero-slope is just for convenience as it avoids to introduce a further parameter in the scaling. However, the inputs used in the proofs are available for non-zero slopes as well.}
\begin{enumerate}
 \item \emph{Droplet case.} In this case one sets $h^0=0$ and further set $\omega(i,j)=0$ whenever $(i,j)\not\in\Z_+^2$. In terms of LPP this is equivalent to take $S_A={(0,0)}$ and $h^0=0$.
 \item \emph{Flat with zero-slope.} This means that we take $h^0=0$.
 \item \emph{Stationary with zero-slope.} Let $\{X_k,Y_k\}_{k\in\Z}$ be i.i.d.\ random variable ${\rm Exp}(1/2)$-distributed. Then define
 \begin{equation}\label{eq1.7}
h^0(x,-x)=\begin{cases}
  \sum_{k=1}^x (X_k-Y_k),& \textrm{for }x\geq 1,\\
  0, & \textrm{for }x=0,\\
  -\sum_{k=x+1}^0 (X_k-Y_k),& \textrm{for }x\leq -1.
 \end{cases}
\end{equation}
 \item \emph{A family of random initial conditions.} We consider the case where for a given $\sigma\geq 0$, $h^0$ is given by (\ref{eq1.7}) multiplied by $\sigma$. Clearly, the cases $\sigma=0$ and $\sigma=1$ correspond to the flat and to the stationary cases.
\end{enumerate}

\begin{rem}
In the setting of TASEP, a random initial condition maps to a LPP starting from a random line. Due to functional slow-decorrelation, the weight $h^0$ should be taken to reflect the first order LPP from a point on the line to its projection onto the antidiagonal. Thus a-priori one could try to start with the random line used in~\cite{CFS16,FO17}, but since in the scaling limit the result is identical to the one of our choice, we did not attempt to use this precise mapping.
\end{rem}

\subsubsection*{Limiting variational formulas}
For $0<\tau\leq 1$, we set\footnote{Throughout the paper we do not write explicitly integer parts.} $E_\tau=(\tau N,\tau N)+(2N)^{2/3} w_\tau (1,-1)$ and define the LPP and its limit as
\begin{equation}\label{eqL}
L^\star_N(\tau)=\frac{L^\star_{S_A\to E_\tau}-4\tau N}{2^{4/3}N^{1/3}},\quad \chi^{\star}(\tau):=\lim_{N\to\infty} L^{\star}_N(\tau),
\end{equation}
where the superscript $\star$ denotes the different configurations, point-to-point ($\bullet$), point-to-line ($\diagdown$), stationary ($\mathcal{B}$) and random ($\sigma$).

The convergence in distribution of the random variables $L^{\star}_N(\tau)$ are well-known. Recall that for LPP we have the identity
\begin{equation}\label{eqLstar}
L^{\star}_{S_A\to E_1}=\max_{u\in\R}\{L^{\star}_{S_A\to I(u)}+L_{I(u)\to E_1}\}
\end{equation}
with
\begin{equation}\label{eqI}
I(u)=(\tau N,\tau N)+u(2N)^{2/3}(1,-1).
\end{equation}
Provided that the limit $N\to\infty$ and $\max_{u\in\R}$ can be exchanged (which is the case in all the cases considered here, see~\cite{CLW16,FO17,CFS16} for related works), the limiting processes can be written in terms of Airy processes as follows.
\begin{enumerate}
\item \emph{Droplet case.} Let ${\cal A}_2$ and $\tilde {\cal A}_2$ be two independent Airy$_2$ processes. Then
\begin{equation}\label{chistep}
\begin{aligned}
 \chi^{\bullet}(\tau)&=\tau^{1/3}\left[\tilde{\cal A}_2(\tfrac{w_\tau}{\tau^{2/3}})-\tfrac{w_\tau^2}{\tau^{4/3}}\right],\\
 \chi^{\bullet}(1)&=\max_{u\in\R}\left\{\tau^{1/3}\left[\tilde{\cal A}_2(\tfrac{u}{\tau^{2/3}})-\tfrac{u^2}{\tau^{4/3}}\right]+(1-\tau)^{1/3}\left[{\cal A}_2\bigl(\tfrac{u-w_1}{(1-\tau)^{2/3}}\bigr) - \tfrac{(u-w_1)^2}{(1-\tau)^{4/3}}\right]\right\},
\end{aligned}
\end{equation}
The Airy$_2$ process has been discovered in a related polynuclear growth model setting~\cite{PS02} (see~\cite{Jo03b} for the case of geometric random variables, or~\cite{BP07} for a two-parameter generalization). Tightness in this setting was shown in~\cite{FO17}, building on the approach of~\cite{CP15b} (while for the geometric case tightness was shown already in~\cite{Jo03b}).
\item \emph{Flat case.} Let ${\cal A}_1$ be an Airy$_1$ process and ${\cal A}_2$ an Airy$_2$ process, independent of each other. Then
\begin{equation}\label{chiflat}
\begin{aligned}
 \chi^{\diagdown}(\tau)&=(2\tau)^{1/3}{\cal A}_{1}(\tfrac{w_\tau}{(2\tau)^{2/3}}),\\
 \chi^{\diagdown}(1)&=\max_{u\in\R} \left\{ (2\tau)^{1/3}{\cal A}_1 (\tfrac{u}{(2\tau)^{2/3}}) + (1-\tau)^{1/3}\left[{\cal A}_2\bigl(\tfrac{u-w_1}{(1-\tau)^{2/3}}\bigr) - \tfrac{(u-w_1)^2}{(1-\tau)^{4/3}}\right]\right\}.
\end{aligned}
\end{equation}
The Airy$_1$ process has been discovered in the framework of the totally asymmetric simple exclusion process~\cite{Sas05,BFPS06}, equivalent through slow-decorrelation to the LPP~\cite{CFP10a,CFP10b,Fer08}.
\item \emph{Stationary case.} Let ${\cal A}_2$ be an Airy$_2$ process and ${\cal A}_{\rm stat}$ an Airy$_{\rm stat}$ process, independent of each other. Then
\begin{equation}\label{chistat}
\begin{aligned}
\chi^{\mathcal{B}}(\tau)&=\tau^{1/3}{\cal A}_{\rm stat}(\tfrac{w_\tau}{\tau^{2/3}}),\\
\chi^{\mathcal{B}}(1)&=\max_{u\in\R}\Big\{\tau^{1/3}{\cal A}_{\rm stat}(\tfrac{u}{\tau^{2/3}})+(1-\tau)^{1/3}\left[{\cal A}_2\bigl(\tfrac{u-w_1}{(1-\tau)^{2/3}}\bigr) - \tfrac{(u-w_1)^2}{(1-\tau)^{4/3}}\right]\Big\}.
\end{aligned}
\end{equation}
The limit process Airy$_{\rm stat}$ (which, in spite of the name, is not stationary) was obtained in~\cite{BFP09}.
\item \emph{Random initial conditions.} For this case, the one-point distribution is given by the following expression\footnote{This was actually proven for the LPP model where instead of the random function on the antidiagonal one has a random line in~\cite{CFS16}, see also~\cite{FO17} for general slope. These works were based on the approach in the geometric random variables case of~\cite{CLW16}. Adapting the proof of~\cite{FO17} to this setting to get the variational formula is straightforward (it is actually even slightly simpler).}
\begin{equation}
\Pb(\chi^\sigma(1)\leq s)=\Pb\left(\max_{u\in\R} \{{\cal A}_2(u)-u^2+\sqrt{2}\sigma B(u)\}\leq s\right),
\end{equation}
where the Airy$_2$ process and the two-sided standard Brownian motion $B$ are independent of each other. Furthermore, we could write formulas similar to the one of the first three cases in terms of \emph{an} Airy sheet~\cite{MQR17}. However uniqueness in law of Airy sheet is so-far not proven~\cite{MQR17,Pi17b}. Therefore we state the convergence of the covariance to the covariance of its limit process only for the other cases. However, the proof could be adapted to the general $\sigma$ as well, once uniqueness of the limit is established.
\end{enumerate}

\subsection{Main results}
\subsubsection*{Convergence of the covariance}
As our first result we give a rigorous proof of the convergence of the covariances.
\begin{thm}\label{Thm1}We have
\begin{equation}\label{EqThm1a}
 \lim_{N\to\infty}\Cov\left(L^{\star}_N(\tau),L^{\star}_N(1)\right)=\Cov\left(\chi^{\star}(\tau),\chi^{\star}(1)\right),
\end{equation}
for $\star\in\{\bullet,\diagdown,\mathcal{B}\}$.
\end{thm}
\begin{remark}
The motivation of this paper is the study of the covariance. However, by inspecting the proof, one sees that one can generalize the proof to get  convergence of any joint moments of $L^\star_N(\tau)$ and $L^\star_N(1)$ without the need to new ideas and bounds.
\end{remark}

For the stationary process ${\cal A}_{\rm stat}(w)\stackrel{(d)}{=}\max_{v\in\R}\{\sqrt{2}B(v)+{\cal A}_2(v)-(v-w)^2\}$ where the Airy$_2$ process, ${\cal A}_2$, and the two-sided standard Brownian motion, $B(v)$, are independent~\cite{QR13}. We denote
\begin{equation}
F_w(s)=\Pb\Big(\max_{v\in\R}\{\sqrt{2}B(v)+{{\cal A}}_2(v)-(v-w)^2\}\leq s\Big)
\end{equation}
and use the notation $\xi_{{\rm stat},w}$ for a random variable distributed according to $F_w$. Due to stationarity one has the property~\cite{PS01,BR00} $\E({\cal A}_{\rm stat}(w))=0$, which implies
\begin{equation}
\Var(\xi_{{\rm stat},w})=\E\Big(\max_{v\in\R}\{\sqrt{2}B(v)+{{\cal A}}_2(v)-(v-w)^2\}\leq s\Big)^2.
\end{equation}

For the stationary case, an exact expression for the covariance has been obtained in~\cite{FS16} for $\tau$ in the entire interval [0,1], in the special case $w_\tau=w_1=0$. For general values of $w_\tau$ and $w_1$, we obtain
\begin{cor}\label{Cor:covstat}
For the stationary LPP, the covariance of the limiting height function for all $\tau\in (0,1)$ can be expressed as
\begin{equation}\label{eq1.13}
\begin{aligned}
\Cov\left(\chi^{\mathcal{B}}(\tau),\chi^{\mathcal{B}}(1)\right)&=\frac{\tau^{2/3}}{2}\Var\left(\xi_{{\rm stat},\tau^{-2/3}w_\tau})\right)+\frac{1}{2}\Var\left(\xi_{{\rm stat},w_1}\right)\\
&-\frac{(1-\tau)^{2/3}}{2}\Var\left(\xi_{{\rm stat},(1-\tau)^{-2/3}(w_1-w_\tau)}\right).
\end{aligned}
\end{equation}
\end{cor}

\subsubsection*{Universal behavior for $\tau\to 1$}
In~\cite{FS16} there is a conjecture on the behaviour of the covariance of the limit process for $\tau\to 1$ for the other initial profiles as well. Our second goal is to provide a proof of such statements together with a rigorous error bound. We also extend the result to all initial conditions 1-4. Recall that for any random variables $X_1,X_2$ it holds
\begin{equation}\label{covariance}
\Cov\left(X_1,X_2\right)=\tfrac{1}{2}\Var\left(X_1\right)+\tfrac{1}{2}\Var\left(X_2\right)-\tfrac{1}{2}\Var(X_2-X_1).
\end{equation}

\begin{thm}\label{thm:covGeneral} Let us scale $w_1=\tilde w_1 (1-\tau)^{2/3}$ and $w_\tau=\tilde w_\tau (1-\tau)^{2/3}$. Then as $\tau\to 1$ we have\footnote{One could also reformulate the result by saying that the error term is $\Or((1-\tau)/\ln(1-\tau))$.}
 \begin{equation}
   \Var\left(\chi^{\star}(\tau)-\chi^{\star}(1)\right)=(1-\tau)^{2/3} \Var\left(\xi_{{\rm stat},\tilde w_1-\tilde w_\tau}\right) +\Or(1-\tau)^{1-\delta},
 \end{equation}
 for any $\delta>0$.
 In particular, by (\ref{covariance}), for $\star=\{\bullet,\diagdown,\mathcal{B}\}$, we can rewrite
  \begin{equation}
 \begin{aligned}
   \Cov\left(\chi^{\star}(\tau),\chi^{\star}(1)\right)&=\frac{1}{2} \Var\left(\xi^{\star}(w_1)\right)+\frac{\tau^{2/3}}{2}\Var\left(\xi^{\star}(w_\tau\tau^{-2/3})\right)\\
   &-\frac{(1-\tau)^{2/3}}{2} \Var\left(\xi_{{\rm stat},\tilde w_1-\tilde w_\tau}\right) +\Or(1-\tau)^{1-\delta}.
\end{aligned}
 \end{equation}
Here $\xi^\bullet(w)+w^2$ (resp.\ $2^{2/3}\xi^\diagdown(w)$) is distributed according to a GUE (resp.\ GOE) Tracy-Widom law and $\xi^{\mathcal{B}}(w)=\xi_{{\rm stat},w}$.
\end{thm}

\subsubsection*{Small $\tau$ behavior for droplet initial conditions}
\begin{thm}\label{thm:covstep0} For point-to-point LPP, let $w_\tau=\hat w_\tau \tau^{2/3}$. Then the covariance of the limiting height function for $\tau\to 0$ can be expressed as
 \begin{equation}
   \Cov\left(\chi^{\bullet}(\tau),\chi^{\bullet}(1)\right)=\tau^{2/3}\E({\cal A}_2(\hat w_\tau) \max_{u\in\R}\{{\cal A}_2(u)-u^2+\sqrt{2}B(u)\})+\Or(\tau^{1-\delta}).
 \end{equation}
\end{thm}

\section{The stationary LPP and its comparison lemmas}\label{sec:stationaryComparison}
As shown in~\cite{BCS06} the stationary situation can be realized in different ways. For the purpose of this paper, we will consider the following situations
\begin{itemize}
\item On $\Z_+^2$: consider the LPP from $S_A=\{(0,0)\}$ with
\begin{equation}\label{stationaryboundary}
 \omega_{i,j}= \begin{cases}
  0 & \textrm{for }i=0, j=0,\\
  {\rm Exp}(1-\rho) & \textrm{for }i\geq 1, j=0,\\
  {\rm Exp}(\rho) & \textrm{for }i=0, j\geq 1,\\
  {\rm Exp}(1) & \textrm{for }i\geq 1, j\geq 1.
 \end{cases}
\end{equation}
This is called stationary LPP with density $\rho$ since the increments of the LPP along horizontal lines are still sums of iid.\ ${\rm Exp}(1-\rho)$ random variables, as a special case of Lemma~4.2 of~\cite{BCS06}. More generically, the increments along a down-right path are sums of independent random variables, ${\rm Exp}(1-\rho)$ for horizontal steps, and $-{\rm Exp}(\rho)$ for vertical steps.
\item Consider $S_A={\cal L}=\{(i,j)\in \Z^2\, |\, i+j=0\}$ and with boundary terms
\begin{equation}\label{stationaryboundaryB}
h^0(x,-x)=\begin{cases}
  \sum_{k=1}^x (X_k-Y_k),& \textrm{for }x\geq 1,\\
  0, & \textrm{for }x=0,\\
  -\sum_{k=x+1}^0 (X_k-Y_k),& \textrm{for }x\leq -1,
 \end{cases}
\end{equation}
where $\{X_k\}_{k\in\Z}$ and $\{Y_k\}_{k\in\Z}$ are independent random variables with \mbox{$X_k\sim {\rm Exp}(1-\rho)$} and $Y_k\sim {\rm Exp}(\rho)$. Then by Lemma~4.2 of~\cite{BCS06} the increments of the LPP in this model are as in the first case.
\end{itemize}
We will call a stationary LPP model either of this two settings, depending on the cases. When we consider the point-to-point problem, we will refer to the stationary case as the first setting, while, when considering the other initial conditions, the stationary LPP will be the second setting.

To prove Theorem~\ref{thm:covGeneral} we are going to use a comparison with the stationary model of density slightly higher or lower than $1/2$. The comparison idea was first used in~\cite{CP15b} and then generalized in~\cite{Pi17}, with applications in~\cite{Pi17b,FGN17,FO17,N17}. For that purpose, we need to introduce the notion of exit point, which is the location where the maximizer of the LPP exits its boundary terms. Let us define it for both stationary settings.
\begin{defin}\label{exitpoint}
\begin{itemize}
\item The exit point for the stationary LPP to $(m,n)$ with boundary \eqref{stationaryboundary} is the last point on the $x$-axis or the $y$-axis of the maximizer ending at $(m,n)$. We introduce the random variable $Z^{\rho}(m,n)\in\mathbb{Z}$ such that, if $Z^{\rho}(m,n)>0$, then the exit point is $(Z^{\rho}(m,n),0)$, and if $Z^{\rho}(m,n)<0$, then the exit point is $(0,-Z^{\rho}(m,n))$.
\item The exit point for the stationary LPP to $(m,n)$ with boundary \eqref{stationaryboundaryB} is the starting point of the maximizer ending at $(m,n)$. We use the notation $\widetilde Z^{\rho}(m,n)\in\mathbb{Z}$ such that the exit point is $(\widetilde Z^{\rho}(m,n),-\widetilde Z^{\rho}(m,n))$.
\item The exit point for the LPP from ${\cal L}$ with initial condition $h^0$ is the starting point of the maximizer ending at $(m,n)$. We use the notation $ Z_{h^0}(m,n)\in\mathbb{Z}$ such that the exit point is $(Z_{h^0}(m,n),-Z_{h^0}(m,n))$. For the random initial condition with parameter $\sigma$, we denote $Z_{h^0}=Z_\sigma$, and for flat initial condition $Z_{h^0}=Z^\diagdown$.
\end{itemize}
\end{defin}

Now we state the two comparison lemmas which we are going to use in the proof of Theorem~\ref{thm:covGeneral}.
\begin{lem}\label{LemmaIncrementsBounds} Denote by $L^\rho$ the LPP~\eqref{stationaryboundary} and $L^{\bullet}$ the LPP in the droplet case.
 Let $0\le m_1\le m_2$ and $n_1\geq n_2\geq 0$. Then if $Z^{\rho}(m_1,n_1)\geq 0$, it holds
 \begin{equation}
  L^\bullet(m_2,n_2)-L^\bullet(m_1,n_1)\leq L^{\rho}(m_2,n_2)-L^{\rho}(m_1,n_1),
 \end{equation}
while, if $Z^{\rho}(m_2,n_2)\leq 0$, then we have
 \begin{equation}
  L^\bullet(m_2,n_2)-L^\bullet(m_1,n_1)\geq L^{\rho}(m_2,n_2)-L^{\rho}(m_1,n_1).
 \end{equation}
\end{lem}
\begin{lem}\label{LemmaIncrementsBoundsB} Denote by $L^\rho$ the LPP~\eqref{stationaryboundaryB} and $L^\star$ be the LPP from $\cal L$ with boundary term $h^0$.
 Let $0\le m_1\le m_2$ and $n_1\geq n_2\geq 0$. Then if $\widetilde Z^{\rho}(m_1,n_1)\geq \widetilde Z_{h^0}(m_2,n_2)$, it holds
 \begin{equation}
  L^\star(m_2,n_2)-L^\star(m_1,n_1)\leq L^{\rho}(m_2,n_2)-L^{\rho}(m_1,n_1),
 \end{equation}
while, if $\widetilde Z^{\rho}(m_2,n_2)\leq \widetilde Z_{h^0}(m_1,n_1)$, then we have
 \begin{equation}
  L^\star(m_2,n_2)-L^\star(m_1,n_1)\geq L^{\rho}(m_2,n_2)-L^{\rho}(m_1,n_1).
 \end{equation}
\end{lem}
For $n_1=n_2$, Lemma~\ref{LemmaIncrementsBounds} is in Lemma~1 of~\cite{CP15b}, while Lemma~\ref{LemmaIncrementsBoundsB} is Lemma~2.1 of~\cite{Pi17}. The generalization to points on a down-right path is straightforward. It was made for instance in the LPP setting~\eqref{stationaryboundaryB} in Lemma~3.5 of~\cite{FGN17}.

\section{Convergence of the covariance}\label{sec:ConvergenceCovariance}

\subsection{Preliminaries and notations}
A law of large number for point-to-point LPP was proven in~\cite{R81}, namely, for large $(m,n)$, $L_{(0,0)\to (m,n)}\approx (\sqrt{m}+\sqrt{n})^2$. From this we can estimate
\begin{equation}\label{eq3}
\begin{aligned}
L^\star_{(0,0)\to E_{\tau}}&\approx 4\tau N-w_\tau^2 \tau^{-1} 2^{4/3}N^{1/3},\\
L^\star_{(0,0)\to I(u)}&\approx 4\tau N -u^2 \tau^{-1} 2^{4/3}(\tau N)^{1/3},\\
L_{I(u)\to E_1}&\approx 4(1-\tau)N -\frac{(u-w_1)^2}{1-\tau}2^{4/3}N^{1/3}.
\end{aligned}
\end{equation}
Denote the rescaled LPP by
\begin{equation}
L^{\star}_N(u,\tau):=\frac{L^\star_{S_A\to I(u)}-4(1-\tau) N}{2^{4/3} N^{1/3}},
\end{equation}
with $I(u)=(\tau N,\tau N)+u(2N)^{2/3}(1,-1)$ and
\begin{equation}
L^{\rm pp}_N(u,\tau):=\frac{L_{I(u)\to E_1}-4(1-\tau) N}{2^{4/3}N^{1/3}},
\end{equation}
where we recall that $E_1=(N,N)+w_1(2N)^{2/3}(1,-1)$. Then, (\ref{eqL}) and (\ref{eqLstar}) become
\begin{equation}
\begin{aligned}
L^\star_N(\tau)&\equiv L^{\star}_N(w_\tau,\tau),\\
L^\star_N(1)&\equiv L^{\star}_N(w_1,1)=\max_{u\in\R} \{L^{\star}_N(u,\tau)+L^{\rm pp}_N(u,\tau)\}.
\end{aligned}
\end{equation}
Furthermore,
\begin{equation}
\lim_{N\to\infty}L^{\rm pp}_N(u,\tau) = (1-\tau)^{1/3}\left[{\cal A}_2\bigl(\tfrac{u-w_1}{(1-\tau)^{2/3}}\bigr) - \tfrac{(u-w_1)^2}{(1-\tau)^{4/3}}\right],
\end{equation}
and
\begin{equation}
\lim_{N\to\infty} L^{\star}_N(u,\tau) = \tau^{1/3} {\cal A}^\star\bigl(\tfrac{u}{\tau^{2/3}}\bigr), \quad \star\in\{\bullet,\diagdown,\mathcal{B}\},
\end{equation}
where
\begin{equation}
{\cal A}^\bullet(u) = \tilde {\cal A}_2(u)-u^2, \quad {\cal A}^\diagdown(u) = 2^{1/3} {\cal A}_1(u 2^{-2/3}), \quad {\cal A}^{\mathcal{B}}(u) = {\cal A}_{\rm stat}(u).
\end{equation}

\subsection{Localization of the maximizer at time $\tau N$}
The maximizer of the process $L^{\star}_N(u,\tau)+L^{\rm pp}_N(u,\tau)$ is confined in the region with $|u|\leq M$ if the following event holds
\begin{equation}\label{omegaM}
\Omega^G_{M}=\Big\{\max_{|u|\leq M}\{L^{\star}_N(u,\tau)+L^{\rm pp}_N(u,\tau)\}>\max_{|u|> M}\{L^{\star}_N(u,\tau)+L^{\rm pp}_N(u,\tau)\}\Big\}.
\end{equation}
Thus we need to estimate $\Pb(\Omega^G_M)$. For any choice of $s\in\R$ we can write
\begin{equation}\label{goodevent}
\begin{aligned}
 \Pb(\Omega^G_M)& \geq \Pb\left(\max_{|u|\leq M}\{L^{\star}_N(u,\tau)+L^{\rm pp}_N(u,\tau)\}>s>\max_{|u|> M}\{L^{\star}_N(u,\tau)+L^{\rm pp}_N(u,\tau)\}\right)\\
 & \geq  1-\Pb(G_M)-\Pb(B_M),
 \end{aligned}
\end{equation}
where we defined
\begin{equation}
\begin{aligned}
G_M&=\{\max_{|u|\leq M}\{L^{\star}_N(u,\tau)+L^{\rm pp}_N(u,\tau)\}\leq s\},\\
B_M&=\{\max_{|u|> M}\{L^{\star}_N(u,\tau)+L^{\rm pp}_N(u,\tau)\}>s\}.
\end{aligned}
\end{equation}

The right side of \eqref{goodevent} is estimated using the following lemma.
\begin{lem} \label{lemma1}Let $s=-M^2 \tilde c$ with $\tilde c=1/(16(1-\tau))$. Then, there exists a finite $M_0$ such that for any $M\geq M_0$
\begin{equation}\label{eq15.1}
 \begin{aligned}
  \Pb\left(G_M\right)\leq Ce^{-cM^2}\\
  \Pb\left(B_M\right)\leq Ce^{-cM^2}
 \end{aligned}
\end{equation}
for some constants $C,c>0$ uniform in $N$.
\end{lem}
As a direct consequence we have the following localization result.
\begin{cor}\label{cor1}For any $M\geq M_0$,
\begin{equation}
 \Pb\left(\mbox{the maximizer of }L^{\star}_N(u,\tau)+L^{\rm pp}_N(u,\tau) \mbox{ passes by } I(u) \mbox{ with } |u|>M\right)\leq 2Ce^{-cM^2}
\end{equation}
uniformly in $N$.
\end{cor}
Denote by
\begin{equation}\label{chiB_M}
\chi^{\star}_M(1)=\max_{|u|\leq M}\Big\{\tau^{1/3}{\cal A}^\star(\tau^{-2/3} u)+(1-\tau)^{1/3}\bigl[{\cal A}_2\bigl( \tfrac{u-w_1}{(1-\tau)^{2/3}}\bigr) - \tfrac{(u-w_1)^2}{(1-\tau)^{4/3}}\bigr]\Big\}
\end{equation}
and recall
\begin{equation}
\chi^\star(\tau)=\tau^{1/3}{\cal A}^\star(\tau^{-2/3} w_\tau).
\end{equation}

\begin{lem}\label{lem2.5}
We have the convergence of joint distributions
\begin{equation}\label{conv_dist_stat}
\begin{aligned}
 \lim_{N\to\infty} &\Pb\left(\max_{|u|\leq M}\{L^{\star}_N(u,\tau)+L^{\rm pp}_N(u,\tau)\}\leq s_1; L^\star_N(w_\tau,\tau)\leq s_2\right)\\
  = &\Pb\left(\chi^{\star}_M(1)\leq s_1; \chi^\star(\tau)\leq s_2\right).
\end{aligned}
\end{equation}
\end{lem}
\begin{proof}
It is enough to have weak convergence of the two rescaled process to the terms in the rhs. As mentioned above, the point-wise convergence have been already proven. So we need tightness in the space of continuous functions of $[-M,M]$. Tightness $L^{\rm pp}_N(u,\tau)$ and $L^\bullet_N(u,\tau)$ can be found in Corollary~4.2 of~\cite{FO17}, for $L^{\mathcal{B}}_N(u,\tau)$ it is a direct a direct consequence of Lemma~4.2 of~\cite{BCS06} and the standard Donsker's theorem. Finally, tightness for $L^\diagdown_N(u,\tau)$ has been established in~\cite{Pi17}.
\end{proof}
\subsubsection{Localization of the process}
Let us prove Lemma~\ref{lemma1} and Corollary~\ref{cor1}.
\begin{proof}[Proof of Lemma~\ref{lemma1}]Recall that to prove this lemma, we take $s=s_0$, with the choice $s_0=-M^2 \tilde c=-M^2/(16(1-\tau))$.

\medskip
\noindent (1) Bound on $\Pb(G_M)$.\\
We have
\begin{equation}\label{eq15.2}
  \begin{aligned}
  \Pb(G_M)= &\Pb\Big(\max_{|u|\leq M}\{L^{\star}_N(u,\tau)+L^{\rm pp}_N(u,\tau)\}\leq s_0\Big)\\
  \leq &\Pb\left( L^{\bullet}_N(0,\tau)+L^{\rm pp}_N(0,\tau)\leq s_0\right)\\
  \leq &\Pb\left( L^{\bullet}_N(0,\tau)\leq s_0/2\right)+\Pb\left(L^{\rm pp}_N(0,\tau)\leq s_0/2\right).
  \end{aligned}
\end{equation}
Now we can use standard estimates on the lower tail of the point-to-point LPP (see Prop.~\ref{PropBoundsPP} in Appendix~\ref{sec:boundstepLPP}) to obtain that \eqref{eq15.2} is bounded by $Ce^{-cM^3}$ uniformly in $N$, for some constants $C,c$.

\medskip
\noindent (2) Bound on $\Pb(B_M)$. \emph{Since similar estimates will be used to derive another result, we add an extra variable $\hat s\geq 0$ in the following computations. The case $\hat s=0$ is the one relevant for the present proof.}\\
We have
\begin{equation}\label{eq15.3}
\begin{aligned}
\Pb(B_M)&=\Pb\Big(\max_{|u|>M}\left\{L^{\star}_N(u,\tau)+L^{\rm pp}_N(u,\tau)\right\}>s_0+\hat s\Big)\\
&\leq \Pb\Big(\max_{|u|>M}\Big\{L^{\star}_N(u,\tau)-\frac{u^2}{2(1-\tau)}\Big\}>\frac{s_0+\hat s}{2}\Big)\\
  &+\Pb\Big(\max_{|u|>M}\Big\{L^{\rm pp}_N(u,\tau)+\frac{u^2}{2(1-\tau)}\Big\}>\frac{s_0+\hat s}{2}\Big).
 \end{aligned}
\end{equation}
We study separately the two terms of \eqref{eq15.3} and rename them $\Pb\left(B^1_M\right)$ and $\Pb\left(B^2_M\right)$ respectively. Remark that the maximum over $u$ is actually a maximum over $M<|u|\leq \Or(N^{1/3})$, since $I(u)$ need to stay in the backward light cone of the end-point $E_1$. We will not write this explicitly all the time.

\noindent
(a) Estimate of $\Pb\left(B^2_M\right)$:
\begin{equation}\label{eq2.23}
 \begin{aligned}
  \Pb\left(B^2_M\right)&=\Pb\left(\max_{|u|>M}\Big\{L^{\rm pp}_N(u,\tau)+\frac{u^2}{2(1-\tau)}\Big\}>\frac{s_0+\hat s}{2}\right)\\
 \end{aligned}
\end{equation}
The bound can be obtained through the decay of the kernel for half-flat initial condition. The bound on the Fredholm determinant and the kernel are given as in Theorem~2.6 and Lemma~2.7 of~\cite{CFS16} to get
\begin{equation}\label{eq24a}
\Pb\left(B^2_M\right) \leq Ce^{-cM^2(1-\tau)^{-4/3}} e^{-\tilde c \hat s}.
\end{equation}
Alternatively, one could adapt the proof of Lemma~4.3 of~\cite{FO17} to get the same result.

\noindent (b) Estimate of $\Pb\left(B^1_M\right)$:
\begin{equation}\label{eq17}
 \Pb\left(B^1_M\right) =\Pb\Big(\max_{|u|>M}\Big\{L^{\star}_N(u,\tau)-\frac{u^2}{2(1-\tau)}\Big\}>\frac{s_0+\hat s}{2}\Big).
\end{equation}
\begin{itemize}
\item Droplet initial condition: for this case, one can estimate it like we made for (\ref{eq2.23}) (with minor changes in the terms depending on $\tau$). However, since $L^\bullet_N(u,\tau)\leq L^\diagdown_N(u,\tau)$, the droplet upper tail is simply bounded by the upper tail of the flat initial condition case.
\item Flat initial condition: the bound is obtained in Lemma~\ref{lemma6} below.
\item Stationary initial condition: the bounds for the maximum over $u>M$ and for $u<-M$ are similar and thus we present the details only for the first one.
\begin{equation}\label{eq17b}
\begin{aligned}
&\Pb\Big(\max_{u>M}\Big\{L^{\star}_N(u,\tau)-\frac{u^2}{2(1-\tau)}\Big\}>\frac{s_0+\hat s}{2}\Big)\\
 \leq&\Pb\Big(\max_{u>M}\Big\{L^{\star}_N(u,\tau)-L^\star_N(M,\tau)-\frac{u^2}{2(1-\tau)}\Big\}>s_0+\frac{\hat s}{4}\Big)\\
 &+\Pb\Big(L^{\star}_N(0,\tau)>-\frac{s_0}{4}+\frac{\hat s}{8}\Big)+\Pb\Big(L^{\star}_N(M,\tau)-L^\star_N(0,\tau)>-\frac{s_0}{4}+\frac{\hat s}{8}\Big).
\end{aligned}
\end{equation}
We study separately the three terms of the last line of \eqref{eq17b}. The first term is bounded using (\ref{eqB5}), the second with (\ref{eqB2}) and the third one with (\ref{eqB4}), with the final result
\begin{equation}\label{eqB1M}
\Pb(B_M^1)\leq C e^{-c M^2-\tilde c \hat s}
\end{equation}
for some $c,\tilde c$ depending on $\tau$, but uniform for all $N$ large enough.
\end{itemize}
\end{proof}

\begin{proof}[Proof of Corollary~\ref{cor1}]
By \eqref{eq15.2}, \eqref{eqB1M}, \eqref{eq24a} we can conclude that
\begin{equation}
 \Pb\left(\max_{|u|\leq M}\left\{L^{\star}_N(u,\tau)+L^{\rm pp}_N(u,\tau)\right\}>\max_{|u|> M}\left\{L^{\star}_N(u,\tau)+L^{\rm pp}_N(u,\tau)\right\}\right) \geq 1-2Ce^{-cM^2},
\end{equation}
which implies that the probability that the maximizer of $L^{\star}_N(u,\tau)+L^{\rm pp}_N(u,\tau)$ passes through $I(u)$ with $|u|>M$ goes to zero as $2Ce^{-cM^2}$, for some constants $C,c>0$.
\end{proof}
Here for simplicity of notation we rename $\hat s$ as $s$.
\begin{lem}\label{lemma6}
For flat initial condition, there exist $N_0,M_0$ large enough such that for all $N\geq N_0$ and $M\geq M_0$ it holds
\begin{equation}\label{eq3.24}
\Pb\Big(\max_{M<|u|<\Or(N^{1/3})}\Big\{L^{\diagdown}_N(u,\tau)-\frac{u^2}{2(1-\tau)}\Big\}>\frac{s_0+ s}{2}\Big) \leq C e^{-c s}e^{-\tilde c M^2},
\end{equation}
for some constants $C,c,\tilde c$ independent of $N$ and $M$.
\end{lem}
\begin{proof}
By symmetry we can consider only the case $u>M$, since the bounds for $u<-M$ are similar. Fix an $\e\in (0,1/6)$. Then,
\begin{equation}\label{eq3.25}
\begin{aligned}
&\Pb\Big(\max_{M<u<\Or(N^{1/3})}\Big\{L^{\diagdown}_N(u,\tau)-\frac{u^2}{2(1-\tau)}\Big\}>\frac{s_0+s}{2}\Big)\\
&\leq \sum_{\ell=1}^{N^\e} \Pb\Big(\max_{u\in [\ell M,(\ell+1)M]}\Big\{L^{\diagdown}_N(u,\tau)-\frac{u^2}{2(1-\tau)}\Big\}>\frac{s_0+s}{2}\Big)\\
&+\sum_{u\in[N^\e,\Or(N^{2/3})]} \Pb\Big(L^{\diagdown}_N(u,\tau)-\frac{u^2}{2(1-\tau)}>\frac{s_0+s}{2}\Big).
\end{aligned}
\end{equation}

Notice that $v\mapsto L^\diagdown_N(u+v,\tau)$ and $v\mapsto L^\diagdown_N(v,\tau)$ have the same law for any $u$. Thus, we can simply bound (using also \eqref{eqD2})
\begin{equation}
\begin{aligned}
\Pb\Big(L^{\diagdown}_N(u,\tau)-\frac{u^2}{2(1-\tau)}>\frac{s_0+s}{2}\Big) &\leq \Pb\Big(L^{\diagdown}_N(0,\tau)>\frac{s}{2}-\frac{M^2}{32(1-\tau)}+\frac{N^{2\e}}{2(1-\tau)}\Big)\\
&\leq C e^{-c s/2+c M^2/(32(1-\tau))-c N^{2\e}/(2(1-\tau))}\\
&\leq C e^{-\tilde c s-\hat c M^2} e^{-c N^{2\e}/(4(1-\tau))},
\end{aligned}
\end{equation}
for some constants $C,c,\tilde c,\hat c$, where the last inequality holds for all $N\geq N_0(M)$. From this it immediately follows that
\begin{equation}
\sum_{u\in[N^\e,\Or(N^{2/3})]} \Pb\Big(L^{\diagdown}_N(u,\tau)-\frac{u^2}{2(1-\tau)}>\frac{s_0+s}{2}\Big)\leq C e^{-\tilde c s-\hat c M^2}.
\end{equation}

Now we evaluate the first term in (\ref{eq3.25}). Using translation-invariance in $u$ we get
\begin{equation}\label{eq3.28bb}
\begin{aligned}
&\Pb\Big(\max_{u\in [\ell M,(\ell+1)M]}\Big\{L^{\diagdown}_N(u,\tau)-\frac{u^2}{2(1-\tau)}\Big\}>\frac{s_0+s}{2}\Big)\\
&\leq \Pb\Big(\max_{u\in [0,M]}L^{\diagdown}_N(u,\tau)>\frac{s}{2}-\frac{M^2}{32(1-\tau)}+\frac{\ell^2 M^2}{2(1-\tau)}\Big)\\
&\leq \Pb\Big(\max_{u\in [0,M]}L^{\diagdown}_N(u,\tau)>\frac{s}{2}+\frac{\ell^2 M^2}{4(1-\tau)}\Big)
\end{aligned}
\end{equation}

\smallskip
(a) The first case is $s\geq N^{2\e}$. We can still just use the union bound and the exponential decay to get
\begin{equation}
\begin{aligned}
(\ref{eq3.28bb})&\leq N^{2/3} \exp\left(-c \frac{\ell^2 M^2}{4(1-\tau)}-c \frac{s}{2}\right)\leq N^{2/3} \exp\left(-c \frac{\ell^2 M^2}{4(1-\tau)}-c \frac{s}{4}-c \frac{N^{2\e}}{4}\right)\\
&\leq \exp(-c_1 \ell^2 M^2-c_2 s),
\end{aligned}
\end{equation}
for some constants $c_1,c_2>0$ and all $N$ large enough.

\smallskip
(b) The second case is $s\in(0,N^{2\e})$. Since also $\ell\leq N^\e$ we have that \mbox{$x:=\frac{s}{2}+\frac{\ell^2 M^2}{8(1-\tau)}=\Or(N^{2\e})\ll N^{1/3}$}. The idea is to bound the process in terms of the increments of a stationary case. For that reason we first need to get a formula including the increments of the rescaled LPP, namely we get
\begin{equation}\label{eq3.29}
\begin{aligned}
(\ref{eq3.28})&\leq \Pb\Big(L^\diagdown_N(0,\tau)>\frac{s}{4}+\frac{\ell^2 M^2}{8(1-\tau)}\Big) \\
&+\Pb\Big(\max_{u\in [0,M]}\Big\{L^{\diagdown}_N(u,\tau)-L^\diagdown_N(0,\tau)\Big\}>\frac{s}{2}+\frac{\ell^2 M^2}{8(1-\tau)}\Big).
\end{aligned}
\end{equation}
For the first term, we just use (\ref{eqD2}) and obtain
\begin{equation}
\Pb\Big(L^\diagdown_N(0,\tau)>\frac{s}{4}+\frac{\ell^2 M^2}{8(1-\tau)}\Big)\leq C e^{-c s/4-c \ell^2 M^2/(8(1-\tau))}.
\end{equation}
The sum of this bound over $\ell\geq 1$ leads to a bound $\tilde C e^{-c s/4-c M^2/(8(1-\tau))}$. For the second term in (\ref{eq3.29}), define $\rho_+=\frac12+\kappa N^{-1/3}$ and the event
\begin{equation}
\Omega_{N,\kappa}=\{\widetilde Z^{\rho_+}(I(0)) > \widetilde Z^\diagdown(I(u)), \textrm{ for all }u\in [0,M]\}.
\end{equation}
On this event, by Lemma~\ref{LemmaIncrementsBoundsB}, we have
\begin{equation}
L^{\diagdown}_N(u,\tau)-L^\diagdown_N(0,\tau)\leq L^{\rho_+}_N(u,\tau)-L^{\rho_+}_N(0,\tau),
\end{equation}
which in turns gives
\begin{equation}
\begin{aligned}
&\Pb\Big(\max_{u\in [0,M]}\Big\{L^{\diagdown}_N(u,\tau)-L^\diagdown_N(0,\tau)\Big\}>\frac{s}{2}+\frac{\ell^2 M^2}{8(1-\tau)}\Big)\\
&\leq \Pb\Big(\max_{u\in [0,M]}\Big\{L^{\rho_+}_N(u,\tau)-L^{\rho_+}_N(0,\tau)\Big\}>\frac{s}{2}+\frac{\ell^2 M^2}{8(1-\tau)}\Big)+\Pb(\Omega_{N,\kappa}^c).
\end{aligned}
\end{equation}
By stationarity of the increments we have
\begin{equation}
L^{\rho_+}_N(u,\tau)-L^{\rho_+}_N(0,\tau)=\frac{1}{2^{4/3}N^{1/3}}\sum_{i=1}^{\lfloor uN^{2/3}\rfloor} Z_i,
\end{equation}
where $Z_i=X_i-Y_i$ with $X_i$'s \iid \mbox{Exp$(1-\rho_+)$} and $Y_i$'s \iid \mbox{Exp$(\rho_+)$} random variables. Since \mbox{$\mathcal M_u=\sum_{i=1}^{\lfloor uN^{2/3}\rfloor}Z_i$} is a submartingale, so it is $\exp(t {\cal M}_u)$ for $t>0$ (at least for $t$ small enough) and we can use Doob's inequality for submartingales,
\begin{equation}\label{eq3.31}
\Pb\left(\max_{u\in[0,M]}\mathcal M_u\geq x\right)\leq \inf_{t\geq0}\frac{\E\bigl[e^{t\mathcal M_M}\bigr]}{e^{tx}}=\inf_{t\geq0}\frac{\E\bigl[e^{tZ_1 }\bigr]^{\lfloor MN^{2/3}\rfloor }}{e^{tx}}.
\end{equation}
For  and $\rho_+=\frac{1}{2}+\kappa N^{-1/3}$. An explicit computation gives, for $\kappa\in (0,x/(2^{5/3}M))$,
\begin{equation}
(\ref{eq3.31}) \leq \exp\left(-\frac{(2^{1/3}x-4 M\kappa)^2}{4M}+\Or(x^4 N^{-2/3};\kappa^4 N^{-2/3})\right).
\end{equation}
Thus, with the choice $\kappa=x/(2^{8/3}M)$, we find
\begin{equation}
\begin{aligned}
&\Pb\Big(\max_{u\in [0,M]}\Big\{L^{\rho_+}_N(u,\tau)-L^{\rho_+}_N(0,\tau)\Big\}>\frac{s}{2}+\frac{\ell^2 M^2}{8(1-\tau)}\Big)\\
&\leq \exp\left(-\frac{x^2}{16 M}\right) \leq \exp\left(-c_1 s^2 - c_2 \ell^4 M^3\right),
\end{aligned}
\end{equation}
for some constants $c_1,c_2>0$ and all $N$ large enough. Summing this bound over $\ell\geq 1$ we get $C e^{-c_1 s^2-c_2 M^3}$ for some constants $C$, completing the proof of (\ref{eq3.24}).
\end{proof}

\begin{lem}\label{lemma7}
Let $\rho_{\pm}=\frac{1}{2}\pm\kappa N^{-1/3}$. Define the event
\begin{equation}\label{ENk}
\begin{aligned}
\Omega_{N,\kappa}=&\left\{\widetilde Z^{\rho_+}(I(0))\geq \widetilde Z^\diagdown(I(u)), \forall u\in[0,M] \right\}\cap \left\{\widetilde Z^{\rho_-}(I(0))\leq \widetilde Z^\diagdown(I(u)), \forall u\in[-M,0] \right\}.
\end{aligned}
\end{equation}
where the exit points are as in Definition~\ref{exitpoint}. Then, for all $N$ large enough and all $\kappa>0$ with $\kappa=o(N^{1/3})$,
\begin{equation}
\Pb(\Omega_{N,\kappa}^c)\leq C e^{-c\kappa^2}.
\end{equation}
\end{lem}
\begin{proof}
We need to estimate the complement of the probabilities of the two terms in (\ref{ENk}), for instance
\begin{equation}\label{eq3.41}
\Pb(\widetilde Z^{\diagdown}(I(u))>\widetilde Z^{\rho_+}(I(0))),\textrm{ for some }u\in[0,M]).
\end{equation}
The estimates are completely analogous, thus we provide the details only for the first one.

Since $\widetilde Z^\diagdown(I(u))\leq \widetilde Z^\diagdown(I(M))$ for all $u\in[0,M]$, we have
\begin{equation}
\begin{aligned}
(\ref{eq3.41})& \leq\Pb(\widetilde Z^{\diagdown}(I(M))>\widetilde Z^{\rho_+}(I(0))) \\
&\leq\Pb(\widetilde Z^{\diagdown}(I(M))>\alpha  (2N)^{2/3})+\Pb(\widetilde Z^{\rho_+}(I(0)))<\alpha (2N)^{2/3}).
\end{aligned}
\end{equation}

By Lemma~4.3 of~\cite{FO17}, we have that $\Pb(\widetilde Z^{\diagdown}(I(M))>\alpha (2N)^{2/3})\leq Ce^{-c\alpha^2}$, for some constants $C,c\in (0,\infty)$. Using stationarity of the increments along the antidiagonal, we have
\begin{equation}
\widetilde Z^{\rho_+}(n-k,n+k)\overset{d}{=}\widetilde Z^{\rho_+}(n,n)-k.
\end{equation}
Thus,
\begin{equation}\label{eq3.43}
\begin{aligned}
\Pb(\widetilde Z^{\rho_+}(I(0))<\alpha (2N)^{2/3}) &= \Pb(\widetilde Z^{\rho_+}(I(0))-\alpha (2N)^{2/3}\leq 0)\\
&=\Pb(\widetilde Z^{\rho_+}(I(-\alpha))<0)=\Pb(Z^{\rho_+}(I(-\alpha))<0).
\end{aligned}
\end{equation}
The last equality follows from the fact that we can construct the two models on the same randomness (define the random variables in the model (\ref{stationaryboundary}) as image of the ones of (\ref{stationaryboundaryB}) by~\cite{BCS06}), for which $\widetilde Z^{\rho_+}(m,n)<0$ iff $Z^{\rho_+}(m,n)<0$ by simple geometric considerations.

Setting $(\gamma^2 n,n)=I(-\alpha)$, and writing $\rho_+=1/(1+\gamma)+\tilde \kappa n^{-1/3}$, we deduce that $\tilde\kappa=\tau^{2/3}\kappa-2^{-4/3}\alpha \tau^{-1/3}+\Or(\kappa N^{-1/3})$. Lemma~2.5 of~\cite{FO17} states\footnote{By inspecting the proof of Lemma~2.5 of~\cite{FO17}, one sees that it actually holds true not only for any given $\kappa$, but also for all $\kappa\in [0,o(n^{1/3})]$.} that if $\tilde\kappa>0$, then $\Pb(Z^{\rho_+}(\gamma^2 n,n)<0)\leq C e^{-c \tilde \kappa^2}$ for some constants $C,c>0$. We choose $\alpha=2^{1/3}\tau \kappa$, which gives $\tilde\kappa=\tfrac12 \kappa \tau^{2/3}(1+\Or(\kappa/N^{1/3}))$. Then, for all $N$ large enough, we obtain
\begin{equation}
(\ref{eq3.43})\leq C e^{-\tilde c \kappa^2}
\end{equation}
for some constants $C,c>0$.
\end{proof}

\subsubsection{Convergence of the covariance}\label{2ndmoments}
To prove Theorem~\ref{Thm1}, first we show that the $N\to\infty$ limit of the covariance of $L^{\star}_N(\tau)$ and $\max_{|u|\leq M}\left\{L^\star_N(u,\tau)+L^{\rm pp}_N(u,\tau)\right\}$ is the covariance of $\chi^{\star}(\tau)$ and $\chi^{\star}_M(1)$, for fixed $M>0$.
Now that we have proved the localization of the process, we need to show that the covariance of $\chi^{\star}(\tau)$ and $\chi^{\star}(1)$ is the $M\to\infty$ limit of the covariance of the LPP restricted to the region $|u|\leq M$.
\begin{prop}\label{prop1}
For any fixed $M>0$,
\begin{equation}
\lim_{N\to\infty}\Cov\left(L^{\star}_N(w_\tau,\tau),\max_{|u|\leq M}\{L^{\star}_N(u,\tau)+L^{\rm pp}_N(u,\tau)\}\right)=\Cov\left(\chi^{\star}(\tau),\chi^{\star}_M(1)\right).
\end{equation}
\end{prop}
\begin{proof}
Let us denote
\begin{equation}
L^\star_{N;M}(1)=\max_{|u|\leq M}\{L^{\star}_N(u,\tau)+L^{\rm pp}_N(u,\tau)\}.
\end{equation}
By Lemma~\ref{conv_dist_stat} we already have the convergence of joint distributions of $L^{\star}_{N;M}(1)$ and $L^\star_N(\tau)\equiv L^\star_N(w_\tau,\tau)$ to $\chi^{\star}_M(1)$ and $\xi^\star(\tau)$. By Cauchy-Schwarz it is enough to show the convergence of the second moments of $L^{\star}_{N;M}(1)$ and $L^{\star}_N(\tau)$.

For a random variable $X_N$ with distribution function $F_N(s)=\Pb(X_N\leq s)$, we can write
\begin{equation}
\E(X_N^2) = 2\int_{\R_+} s (1-F_N(s))ds-2\int_{\R_-} s F_N(s)ds.
\end{equation}
If we know that $X_N\to X$ in distribution, to show convergence of the second moment we need only to find $g(s)$ independent of $N$ such that $1-F_N(s)\leq g(s)$ for $s\geq 0$, $F_N(s)\leq g(s)$ for $s<0$ and that $g\in L^1(\R)$. Then dominated convergence allows to take the limit in the integrals and obtain $\E(X_N^2)\to \E(X^2)$. Thus our task is to find such bounds. Since $F_N(s)\in [0,1]$, it is enough to get bounds for the tails, i.e., a bound for $1-F_N(s)$ for $s\geq s_0$ and for $F_N(s)$ for $s\leq -s_0$ for some $s_0$.

\medskip \noindent (1) $\lim_{N\to\infty} \E[(L^{\star}_N(\tau))^2] =\E[(\chi^{\star}(\tau))^2]$.
\begin{itemize}
\item bound on lower tails: due to $L^\star_N(\tau)\geq L^\bullet_N(\tau)$, we can use for all cases the lower bound for the droplet initial condition, which is in Proposition~\ref{PropBoundsPP} (by appropriate change of variables).
\item bound on upper tails: (a) for the droplet initial condition, this is in Proposition~\ref{PropBoundsPP}, (b) for the flat initial condition, this is given in Proposition~\ref{propBoundflat}, (c) for the stationary initial condition\footnote{For stationary initial condition, the convergence of all moments was already proven in~\cite{BFP12}.}, we have
    \begin{equation}
    \Pb(L^{\mathcal{B}}_N(\tau)\leq s)\leq \Pb(L^\mathcal{B}_N(0)\leq s/2)+\Pb(L^\mathcal{B}_N(\tau)-L^\mathcal{B}_N(0)\leq s/2).
    \end{equation}
    The first term is bounded using Proposition~\ref{propStatBounds}, while the second using Proposition~\ref{lemmaIncrementsStat}.
\end{itemize}
In all cases we have at least exponential decay of the both the upper and lower tails. This implies the convergence of the second moment as well.

\medskip \noindent (2) $\lim_{N\to\infty} \E[(L^{\star}_{N;M}(1))^2] =\E[(\chi^{\star}_M(1))^2]$.
\begin{itemize}
\item bound on lower tails: we have
\begin{equation}\label{eq3.28}
\begin{aligned}
\Pb(L^{\star}_{N;M}(1)\leq s)& \leq \Pb(L^{\star}_{N}(0,\tau)\leq s/2)+\Pb(L^{\rm pp}_N(0,\tau)\leq s/2)\\
& \leq \Pb(L^{\bullet}_{N}(0,\tau)\leq s/2)+\Pb(L^{\rm pp}_N(0,\tau)\leq s/2)\leq C e^{-c |s|^{3/2}}
\end{aligned}
\end{equation}
by Proposition~\ref{PropBoundsPP}.
\item bound on upper tails: we have $L^\star_{N;M}(1)\leq L^\star_{N}(1)$ and thus by the estimates used already in part (1) we have
\begin{equation}
\Pb(L^\star_{N}(1)\geq s)\leq C e^{-c s}.
\end{equation}
\end{itemize}
These bounds implies convergence of the second moment as well.
\end{proof}

What remains to prove Theorem~\ref{Thm1} is a control on the contribution to the covariance from the events when the maximizer passed by $I(u)$ for some $|u|>M$.
We have the decomposition
\begin{equation}
\begin{aligned}
&\Cov\left(L^{\star}_N(\tau),L^{\star}_N(1)\right)\\
&=\Cov\left(L^{\star}_N(\tau),L^{\star}_{N;M}(1)\right)+\Cov\left(L^{\star}_N(\tau),L^{\star}_N(1)-L^\star_{N;M}(1)\right).
\end{aligned}
\end{equation}
Given the convergence of the second moments for fixed $M$ by Proposition~\ref{prop1}, there is one term left to study:
 \begin{equation}\label{eq24}
 \begin{aligned}
  &|\Cov\left(L^{\star}_N(\tau),L^{\star}_N(1)-L^\star_{N;M}(1)\right)|\\
  &=|\E\left[L^{\star}_N(\tau)\left(L^{\star}_N(1)-L^\star_{N;M}(1)\right)\right]-\E\left[L^{\star}_N(\tau)\right]\E\left[L^{\star}_N(1)-L^\star_{N;M}(1)\right]|\\
  &\leq 2 \left(\E\left[(L^{\star}_N(\tau))^2\right]\E\left[(L^{\star}_N(1)-L^\star_{N;M}(1))^2\right]\right)^{1/2},
 \end{aligned}
 \end{equation}
where we used Cauchy-Schwarz to control the second term.

\begin{lem}\label{lemma2} For any $M>0$,
\begin{equation}\label{eq3.32}
 \lim_{N\to\infty}\E\left[(L^{\star}_N(1)-L^\star_{N;M}(1))^2\right]\leq Ce^{-cM^2}
 \end{equation}
as well as
\begin{equation}\label{eq3.32b}
\E[(\chi^\star(1)-\chi^\star_M(1))^2]\leq C e^{-c M^2}.
\end{equation}
where $C, c>0$ are positive constants, uniformly in $N$.
\end{lem}

\begin{proof}
Let us denote
\begin{equation}
L^\star_{N;M^c}(1)=\max_{|u|> M}\{L^{\star}_N(u,\tau)+L^{\rm pp}_N(u,\tau)\}.
\end{equation}
Since $L^{\star}_N(1)=\max\{L^{\star}_{N;M}(1),L^\star_{N;M^c}(1)\}$,
we can write
\begin{equation}
L^{\star}_N(1)-L^{\star}_{N;M}(1)=\max\{0,L^\star_{N;M^c}(1)-L^\star_{N;M}(1)\}.
\end{equation}
Integrating by parts, we obtain
\begin{equation}\label{eq2.43}
  \E\left[\left(L^{\star}_N(1)-L^{\star}_{N;M}(1)\right)^2\right]=2\int_{\R_+}s\Pb\left(L^\star_{N;M^c}(1)-L^\star_{N;M}(1)>s\right).
\end{equation}
The probability in the r.h.s.~of~(\ref{eq2.43}) can be bounded as
\begin{equation}
\Pb\left(L^\star_{N;M^c}(1)-L^\star_{N;M}(1)>s\right)
\leq\Pb\left(L^\star_{N;M^c}(1)>\frac{s+\alpha}{2}\right)+\Pb\left(L^\star_{N;M}(1)\leq \frac{\alpha-s}{2}\right)
\end{equation}
for any choice of $\alpha$.

(a) We use the first inequality in (\ref{eq3.28bb}) and obtain
\begin{equation}
\Pb\left(L^\star_{N;M}(1)< \frac{\alpha-s}{2}\right) \leq \Pb\left(L^\bullet_{N}(0,\tau)< \frac{\alpha-s}{4}\right)+ \Pb\left(L^{\rm pp}(0,\tau)< \frac{\alpha-s}{4}\right).
\end{equation}
For any $\alpha<0$ and $s\geq 0$, by Proposition~\ref{PropBoundsPP} we get
\begin{equation}\label{eq2.46}
\Pb\left(L^\star_{N;M}(1)< \frac{\alpha-s}{2}\right) \leq C e^{-c(s-\alpha)^{3/2}}
\end{equation}
for some constants $C,c$. Thus it is enough to choose $\alpha=-\gamma M^2$ for some $\gamma>0$.

(b) Next we bound $\Pb\left(L^\star_{N;M^c}(1)>\frac{s+\alpha}{2}\right)$. Choosing $\alpha=-\frac{M^2}{16(1-\tau)}$ and using the bounds for $\Pb(B_M)$ in the proof of Lemma~\ref{lemma1}, we obtain
\begin{equation}\label{eq231}
\Pb\left(L^\star_{N;M^c}(1)>\frac{s+\alpha}{2}\right)\leq C e^{-cM^2} e^{-\tilde c s}.
\end{equation}

Plugging the bounds (\ref{eq2.46}) and (\ref{eq231}) into (\ref{eq2.43}) leads (\ref{eq3.32}).

Finally, (\ref{eq3.32b}) is proven as follows. By dominated convergence we have that
\begin{equation}
\begin{aligned}
\E[(\chi^\star(1)-\chi^\star_M(1))^2] & = 2\int_{\R_+}s\Pb\left(\chi^\star_{M^c}(1)-\chi^\star_{M}(1)>s\right) \\
& = \lim_{N\to\infty} 2\int_{\R_+}s\Pb\left(L^\star_{N;M^c}(1)-L^\star_{N;M}(1)>s\right) \leq C e^{-c M^2}
\end{aligned}
\end{equation}
where the last inequality follows from (\ref{eq3.32}).
\end{proof}

\begin{proof}[Proof of Theorem~\ref{Thm1}]
We have
\begin{equation}
\begin{aligned}
&\lim_{N\to\infty} \Cov\left(L^{\star}_N(\tau),L^{\star}_N(1)\right) \\
&= \lim_{M\to\infty}\lim_{N\to\infty}\Cov\left(L^{\star}_N(\tau),L^{\star}_{N,M}(1)\right)+
\lim_{M\to\infty}\lim_{N\to\infty}\Cov\left(L^{\star}_N(\tau),L^\star_N(1)-L^{\star}_{N,M}(1)\right).
\end{aligned}
\end{equation}
By Proposition~\ref{prop1}, the first term equals $\Cov(\chi^\star(\tau),\chi^\star_M(1))$. By Lemma~\ref{lemma2}, the second term is $0$. Thus what remains is to show that
\begin{equation}
\Cov(\chi^\star(\tau),\chi^\star(1)) = \lim_{M\to\infty} \Cov(\chi^\star(\tau),\chi^\star_M(1)).
\end{equation}
This is obtained once we prove that
\begin{equation}
\lim_{M\to\infty}\E[(\chi^\star(1)-\chi^\star_M(1))^2]=0,
\end{equation}
which is also part of Lemma~\ref{lemma2}.
\end{proof}

\subsection{Formula for the stationary case}\label{sec:formulaStatCase}
Now we prove the claimed formula for the stationary case. It follows by a simple computation using the result of Theorem~\ref{Thm1} for the stationary case and the identity (\ref{covariance}).

\begin{proof}[Proof of Corollary~\ref{Cor:covstat}]
Setting $X_1=\chi^{\mathcal{B}}(\tau)$ and $X_2=\chi^{\mathcal{B}}(1)$ in (\ref{covariance}) we get
\begin{equation}
\Cov\bigl(\chi^{\mathcal{B}}(\tau),\chi^{\mathcal{B}}(1)\bigr) = \tfrac12 \Var(\chi^{\mathcal{B}}(\tau))+\tfrac12 \Var(\chi^{\mathcal{B}}(1))-\tfrac12 \Var(\chi^{\mathcal{B}}(1)-\chi^{\mathcal{B}}(\tau)).
\end{equation}
The first two terms in (\ref{eq1.13}) are an immediate consequence of the convergence of moments, see proof of Proposition~\ref{prop1}. For the last term, we have
\begin{equation}
\begin{aligned}
\chi^{\mathcal{B}}(1)-\chi^{\mathcal{B}}(\tau) = \max_{u\in\R}&\left\{(1-\tau)^{1/3}\left[{\cal A}_2\bigl(\tfrac{u-w_1}{(1-\tau)^{2/3}}\bigr) - \tfrac{(u-w_1)^2}{(1-\tau)^{4/3}}\right]\right.\\
&\left.+\tau^{1/3}[{\cal A}_{\rm stat}(\tau^{-2/3} u)-{\cal A}_{\rm stat}(\tau^{-2/3} w_\tau)]\right\}.
\end{aligned}
\end{equation}
Changing the variable $u=w_\tau+z (1-\tau)^{2/3}$, and calling $\xi=\frac{w_1-w_\tau}{(1-\tau)^{2/3}}$, it gives
\begin{equation}
\begin{aligned}
\chi^{\mathcal{B}}(1)-\chi^{\mathcal{B}}(\tau) = &(1-\tau)^{1/3}\max_{z\in\R}\left\{{\cal A}_2\bigl(z-\xi\bigr) - (z-\xi)^2\right.\\
&\left.+\frac{\tau^{1/3}}{(1-\tau)^{1/3}}\left[{\cal A}_{\rm stat}(\tau^{-2/3} (w_\tau+z(1-\tau)^{2/3}))-{\cal A}_{\rm stat}(\tau^{-2/3} w_\tau)\right]\right\}.
\end{aligned}
\end{equation}
Next we use the facts: (a) ${\cal A}_2(z-\xi)\stackrel{(d)}{=}{\cal A}_2(z)$, (b) ${\cal A}_{\rm stat}(a+x)-{\cal A}_{\rm stat}(a)\stackrel{(d)}{=}\sqrt{2} B(x)$ with $B$ a standard Brownian motion, and (c) the scaling of Brownian motion, to get
\begin{equation}
\chi^{\mathcal{B}}(1)-\chi^{\mathcal{B}}(\tau) \stackrel{(d)}{=} (1-\tau)^{1/3}\max_{z\in\R}\left\{\sqrt{2} B(z)+{\cal A}_2(z)-(z-\xi)^2\right\} \stackrel{(d)}{=}{\cal A}_{\rm stat}(\xi).
\end{equation}
\end{proof}

\section{Behavior around $\tau=1$}\label{sec:CloseTimeBehavior}

What we have to prove is
\begin{equation}\label{eqlemma3}
  \Var\left[\chi^{\star}(1)-\chi^{\star}(\tau)\right]=(1-\tau)^{2/3}\Var\left(\xi_{{\rm stat},\tilde w}\right)+\Or(1-\tau)^{1-\delta},
\end{equation}
as $\tau\to 1$ for all the initial conditions. Clearly the flat and stationary are special case of the more generic random initial conditions.
Define
\begin{equation}\label{chistep1M}
\begin{aligned}
\chi^\star_M&=\lim_{N\to\infty} \max_{|u|\leq (1-\tau)^{2/3} M} (L^\star_N(u,\tau)+L^{\rm pp}_N(u,\tau))\\
&=\lim_{N\to\infty} \max_{|v|\leq M} (L^\star_N((1-\tau)^{2/3} v,\tau)+L^{\rm pp}_N((1-\tau)^{2/3} v,\tau)).
\end{aligned}
\end{equation}
In particular, for droplet, flat, stationary initial conditions, we have
\begin{equation}
\chi^\star_M = (1-\tau)^{1/3} \max_{|v|\leq M} \bigg(
\left(\tfrac{\tau}{1-\tau}\right)^{1/3} {\cal A}^\star\left(v\left(\tfrac{1-\tau}{\tau}\right)^{2/3}\right)+{\cal A}_2(v-\tilde w_1)-(v-\tilde w_1)^2\bigg),
\end{equation}
with $\cal A^\star$ being the Airy$_2$, Airy$_1$ or Airy$_{\rm stat}$ process for $\star=\bullet,\diagdown,{\cal B}$ respectively. Also, recall the notation
\begin{equation}
\chi^\star(\tau)=\lim_{N\to\infty} L^\star_N((1-\tau)^{2/3} \tilde w_\tau,\tau) = (1-\tau)^{1/3} \left(\tfrac{\tau}{1-\tau}\right)^{1/3} {\cal A}^\star\left(\tilde v(\tau)\left(\tfrac{1-\tau}{\tau}\right)^{2/3}\right).
\end{equation}

On short scales, ${\cal A}^\star$ is expected to behave similar to the stationary state, which is a two-sided Brownian motion with diffusion coefficient $2$. Since the Airy$_2$ process is stationary, for $\tau\to 1$, $\chi^\star_M-\chi^\star(\tau)$ should be close to the following expression
\begin{equation}
\xi_{M,\tilde w_\tau,\tilde w_1}:=(1-\tau)^{1/3}\max_{|v|\leq M}\left\{\sqrt{2}B(v-\tilde w_\tau)+{{\cal A}}_2(v)-(v-\tilde w_1)^2\right\}.
\end{equation}
In this proof we set
\begin{equation}\tilde w=\tilde w_1-\tilde w_\tau.
\end{equation}
For $M=\infty$, replacing $v-\tilde w_\tau\to \tilde v$ and using the stationarity of ${\cal A}_2$ we obtain
\begin{equation}
\xi_{\infty,\tilde w_\tau,\tilde w_1}\stackrel{(d)}{=}(1-\tau)^{1/3}\max_{\tilde v\in\R}\{\sqrt{2}B(\tilde v)+{{\cal A}}_2(\tilde v)-(\tilde v-\tilde w)^2\}=\xi_{{\rm stat},\tilde w}.
\end{equation}
Note that in distribution
\begin{equation}
(1-\tau)^{1/3}(\xi_{\rm GUE}-\tilde w^2)\leq \xi_{M,\tilde w_\tau,\tilde w_1}\leq (1-\tau)^{1/3}\xi_{{\rm stat},\tilde w}
\end{equation}
and therefore we know that the $m$th moment of $\xi_{M,\tilde w_\tau,\tilde w_1}$ is finite and of order $(1-\tau)^{m/3}$.

To control the error term, the idea is to take $M$ depending on $\tau$ such that $M\to\infty$ as $\tau$ goes to 1. Then the task is to prove that the difference between the second moment of $\chi^{\star}_M(1)-\chi^{\star}(\tau)$ and the second moment of $\xi_{M,\tilde w_\tau,\tilde w_1}$ goes to zero as $\tau\to1$.
\begin{lem}\label{lemma4}
 Let $M=\frac{1}{(1-\tau)^{\delta/2}}$ with $\delta>0$. Then
  \begin{equation}
  \Big|\E\left[(\chi^{\star}_M(1)-\chi^{\star}(\tau))^2\right]-\E\left[\xi_{M,\tilde w_\tau,\tilde w_1}^2\right]\Big|=\Or(1-\tau)^{1-\delta}.
 \end{equation}
\end{lem}
We need to control how close the increments of the process over distances of order $(1-\tau)^{2/3}$ at time $\tau$ are with respect to the increments of Brownian motion.

We present a short technical lemma that will be used in the proof of Lemma~\ref{lemma4}.
Recall Definition~\ref{exitpoint} of the exit point for a LPP with boundary conditions \eqref{stationaryboundary} (for the droplet case) or \eqref{stationaryboundaryB} (for the random case) and define $\rho^{\pm}=\frac{1}{2}\pm\kappa N^{-1/3}$. Denote by $L^{\rm \rho^\pm}$ its associated LPP.

\begin{lem}\label{lemma5}
There is an event $\Omega_\kappa$ with $\Pb(\Omega_\kappa)\geq 1-C \exp\left(-c\tilde\kappa^2\right)$, with\\
\emph{(a) for droplet initial condition,} $\tilde \kappa=\kappa-\frac{M (1-\tau)^{2/3}}{2^{4/3}\tau}$,\\
\emph{(b) for random initial condition,} $\tilde \kappa=\kappa-2\frac{M (1-\tau)^{2/3}}{2^{4/3}\tau}$,\\
and constants $C,c,M_0\in (0,\infty)$, such that on $\Omega_\kappa$ the inequalities
\begin{equation}\label{eq70}
 \xi_{M,\tilde w_\tau,\tilde w_1}-\e_0\leq \chi^{\star}_M(1)-\chi^{\star}(\tau)\leq \xi_{M,\tilde w_\tau,\tilde w_1}+\e_0,
\end{equation}
hold in distribution, for all $M\geq M_0$ with $\e_0=\Or(\kappa M (1-\tau)^{2/3})$, under the condition $\tilde\kappa>0$.
\end{lem}
\begin{proof}
Let us define
\begin{equation}
\Delta^{\bullet}_N(u)=\frac{L_{(0,0)\to I(u)}-L_{(0,0)\to E_{\tau}}}{2^{4/3}N^{1/3}}, \quad
\Delta^{\sigma}_N(u)=\frac{L^\sigma_{{\cal L}\to I(u)}-L^\sigma_{{\cal L}\to E_{\tau}}}{2^{4/3}N^{1/3}},
\end{equation}
and recall the definitions
\begin{equation}
\begin{aligned}
L^{\bullet}_N(u,\tau)=&\frac{L_{(0,0)\to I(u)}-4\tau N}{2^{4/3}N^{1/3}},\quad L^{\sigma}_N(u,\tau)=&\frac{L^\sigma_{{\cal L}\to I(u)}-4\tau N}{2^{4/3}N^{1/3}},\\
L^{\rm pp}_N(u,\tau)=&\frac{L_{I(u)\to E_1}-4(1-\tau) N}{2^{4/3}N^{1/3}}.
\end{aligned}
\end{equation}
Then we have
\begin{equation}
L^{\star}_N(u,\tau)=L^{\star}_N(0,\tau)+\Delta^{\star}_N(u).
\end{equation}
Also, recall that we will use the notation $u=v(1-\tau)^{2/3}$ and $\tilde M=M(1-\tau)^{2/3}$.

Define the event
\begin{equation}
\Omega_{N,\kappa}=\left\{
\begin{array}{ll}
\{Z^{\rho_+}(I(-\tilde M))>0,Z^{\rho-}(I(\tilde M))<0\},& \textrm{ for droplet IC},\\
\{\widetilde Z^{\rho_+}(I(-\tilde M))>Z_\sigma(I(\tilde M)),\widetilde Z^{\rho-}(I(\tilde M))<Z_\sigma(I(-\tilde M))\},& \textrm{ for random IC}.
\end{array}
\right.
\end{equation}
Then, on the event $\Omega_{N,\kappa}$ we can bound $\Delta^{\star}_N(u)$ with the increments of the stationary LPP with density $\rho_{\pm}$, defined as
\begin{equation}\label{eqBM}
 B^{\pm}(u)=\frac{L^{\rho_{\pm}}(I(u))-L^{\rho_{\pm}}(I(0))-m_{\rho_{\pm}}u (2N)^{2/3}}{2^{4/3}N^{1/3}},
\end{equation}
where $m_{\rho^{\pm}}=\frac{1}{1-\rho^\pm}-\frac{1}{\rho^\pm}=8\kappa N^{-1/3}+\Or(N^{-1})$. Indeed a minimal modification of Lemma~\ref{LemmaIncrementsBounds} implies, for $-\tilde M\leq w<u\leq \tilde M$,
\begin{equation}\label{eq54}
\begin{aligned}
  \left[B^-(u)-B^-(w)\right]-4(u-w)\kappa &\leq L^{\star}_N(u,\tau)-L^{\star}_N(w,\tau)\\
&\leq \left[B^+(u)-B^+(w)\right]+4(u-w)\kappa
\end{aligned}
\end{equation}
for $N$ large enough. Furthermore, $\Var(B^\pm(u))=u 2^{1/2}(1+\Or(N^{-2/3}))$  and $B^\pm(0)=0$.
Thus by Donsker's theorem, $\lim_{N\to\infty} B^\pm(u)=\sqrt{2}B(u)$, with $B(u)$ a standard two-sided Brownian motion in the space of continuous functions on bounded sets.

Recall that
\begin{equation}
\chi^{\star}_M(1)-\chi^{\star}(\tau) = \lim_{N\to\infty} \max_{|v|\leq M} \Big\{L^\star_N\big(v(1-\tau)^{2/3},\tau\big)-L^\star_N(\tilde w_\tau(1-\tau)^{2/3},\tau)+L^{\rm pp}_N(v(1-\tau)^{2/3},\tau)\Big\}
\end{equation}
and also that $v\mapsto L^{\rm pp}_N(v(1-\tau)^{2/3},\tau)$ converges weakly to $(1-\tau)^{1/3}[{\cal A}_2(v)-(v-\tilde w_1)^2]$. Thus, taking the $N\to\infty$ limit and using the inequalities (\ref{eq54}) we obtain
\begin{multline}
 \Pb\left(\chi^{\star}_M(1)-\chi^{\star}(\tau)\leq (1-\tau)^{1/3}s\right) \\
\leq \Pb\left(\max_{|v|\leq M}\left\{ \sqrt2 (B(v)-B(\tilde w_\tau)+{\cal A}_2(v)-(v-\tilde w_1)^2 -4\kappa (v-\tilde w_\tau)(1-\tau)^{1/3}\right\}\leq s\right).
\end{multline}
Denoting $\e=\max_{|v|\leq M} |4\kappa (v-w)(1-\tau)^{1/3}|=6\kappa M (1-\tau)^{1/3}$ we obtain
\begin{equation}
\begin{aligned}
& \Pb\left(\chi^{\star}_M(1)-\chi^{\star}(\tau)\leq (1-\tau)^{1/3}s\right) \\
&\leq \Pb\left(\max_{|v|\leq M}\left\{ \sqrt2 B(v-\tilde w_\tau)+{\cal A}_2(v)-(v-\tilde w_1)^2\right\}\leq s+\e\right)\\
&= \Pb\left(\xi_{M,\tilde w_\tau,\tilde w_1}\leq (1-\tau)^{1/3}s+\e_0\right)
\end{aligned}
\end{equation}
with $\e_0=(1-\tau)^{1/3} \e$.

Similarly for the lower bound we get
\begin{equation}
 \Pb\left(\chi^{\star}_M(1)-\chi^{\star}(\tau)> (1-\tau)^{1/3}s\right)
\leq \Pb\left(\xi_{M,\tilde w_\tau,\tilde w_1}> (1-\tau)^{1/3}s+\e_0\right).
\end{equation}

To conclude the proof, we need to estimate $\Pb(\Omega_{N,\kappa})$.\\
\emph{(a) Droplet initial condition:} For this case, we apply Lemma~2.5 of~\cite{FO17}. To estimate $\Pb(Z^{\rho_\pm}(I(\mp\tilde M))>0)$, we need to set $I(\mp\tilde M)=(\gamma^2 n,n)$. This gives
\begin{equation}
\rho_\pm=\frac12\pm\frac{\tilde M}{2^{4/3}\tau N^{1/3}}\pm\frac{\tilde \kappa}{\tau^{2/3}N^{1/3}}.
\end{equation}
Then, Lemma~2.5 of~\cite{FO17} gives
\begin{equation}
\Pb(Z^{\rho_\pm}(I(\mp\tilde M))>0) \geq 1-C e^{-c\tilde \kappa^2} = 1-C e^{-c (\tau^{2/3}\kappa-2^{-4/3}\tilde M\tau^{-1/3})^2}.
\end{equation}
The estimates are uniform for all $N$ large enough. Renaming $c\tau^{4/3}$ as a new constant $c$, and $2C$ by $C$, we get
\begin{equation}
\Pb(\Omega_{N,\kappa})\geq 1-C \exp\left(-c\Big(\kappa-\frac{M (1-\tau)^{2/3}}{2^{4/3}\tau}\Big)^2\right).
\end{equation}

\noindent \emph{(b) Random initial condition:}
We derive a bound only for $\Pb(\widetilde Z^{\rho_+}(I(-\tilde M))<Z_\sigma(I(\tilde M)))$, since bounding $\Pb(\widetilde Z^{\rho-}(I(\tilde M))<Z_\sigma(I(-\tilde M)))$ is completely analogue.

The probability we want to bound is smaller than
\begin{equation}
\Pb(\widetilde Z^{\rho_+}(I(-\tilde M))\leq \alpha (2N)^{2/3})+\Pb(Z_\sigma(I(\tilde M))> \alpha (2N)^{2/3}),
\end{equation}
and we choose $\alpha=2^{1/3}\tau \kappa$.
Exactly as in (\ref{eq3.43}), we have
\begin{equation}
\Pb(\widetilde Z^{\rho_+}(I(-\tilde M))\leq \alpha (2N)^{2/3}) = \Pb(Z^{\rho_+}(I(-\tilde M-\alpha))<0)\leq C e^{-c\tilde \kappa^2}
\end{equation}
with $\tilde\kappa=2\tau^{2/3} \left(\kappa-\frac{(1-\tau)^{2/3}M}{2^{1/3}\tau}\right)$, provided $\tilde\kappa>0$.

Now we bound $\Pb(Z_\sigma(I(\tilde M))> \alpha (2N)^{2/3})$. Let $J(v)=v(2N)^{2/3}(1,-1)$, define the scaled variables
\begin{equation}
L_N(v)=\frac{L_{J(v)\to I(\tilde M)}-4(1-\tau) N}{2^{4/3}N^{1/3}},\quad {\cal W}_N(v)=\frac{h^0(J(v))}{2^{4/3} N^{1/3}}.
\end{equation}
Then,
\begin{equation}\label{eq4.27}
\begin{aligned}
\Pb(Z_\sigma(I(\tilde M))> \alpha (2N)^{2/3}) &\leq \Pb\left(\max_{v\leq \alpha} (L_N(v)+{\cal W}_N(v))\leq -s\right)\\
&+\Pb\left(\max_{v> \alpha} (L_N(v)+{\cal W}_N(v))\geq -s\right).
\end{aligned}
\end{equation}
Since $L_N(v)\sim -(v-\tilde M)^2/(\tau)$, we choose $s=(\alpha-\tilde M)^2/4$.

The first term in (\ref{eq4.27}) is bounded by
\begin{equation}
\Pb\left(L_N(\alpha)+{\cal W}_N(\alpha)\leq -s\right)\leq
\Pb\left(L_N(\alpha)\leq -\tfrac32 s\right) +\Pb\left({\cal W}_N(\alpha)\leq \tfrac12 s\right).
\end{equation}
The first term bounded by $C_1 e^{-c_2 (\alpha-\tilde M)^3}$ by (\ref{eqA4}). Since ${\cal W}_N$ is a (rescaled) sum of iid.\ random variables, we can use the the exponential Chebyshev's inequality (see e.g.\ the proof of (\ref{eqB4})) and obtain a bound $C_2 e^{-c_2 (\alpha-\tilde M)^4/\alpha}$.

The second term in (\ref{eq4.27}) is bounded by
\begin{equation}
\Pb\left(\max_{v\geq \alpha} \left(L_N(v)+\tfrac1{2}(v-\tilde M)^2\right)\geq -\tfrac12 s\right)+
\Pb\left(\max_{v\geq \alpha} \left({\cal W}_N(v)-\tfrac1{2}(v-\tilde M)^2\right)\geq -\tfrac12 s\right).
\end{equation}
The first term is estimated similarly to (\ref{eq2.23}) and leads to a bound $C_3 e^{-c_3 (\alpha-\tilde M)^2}$. The second term is bounded using Doob's maximal inequality (see e.g.\ the proof of (\ref{eqB5}) leading to a bound $C_4 e^{-c_4 (\alpha-\tilde M)^4/\alpha}$).

Combining these bounds we get $\Pb(Z_\sigma(I(\tilde M))> \alpha (2N)^{2/3})\leq C e^{-c \tilde\kappa^2}$, provided $\tilde\kappa>0$, for some constants $C,c\in (0,\infty)$ uniformly for all $\tau$ in a compact subset of $(0,1]$. Up to renaming $c\tau^{4/3}$ to $c$ and the constant $2C$ to $C$ we get the claimed result.
\end{proof}

Now we can prove Lemma~\ref{lemma4}.
\begin{proof}[Proof of Lemma~\ref{lemma4}]
By Lemma~\ref{lemma5} we have, on a event $\Omega_\kappa$ with $\Pb(\Omega^c_\kappa)\leq C e^{-c \tilde\kappa^2}$, with
\begin{equation}
\tilde\kappa=\kappa-\frac{M (1-\tau)^{2/3}}{2^{4/3}\tau}.
\end{equation}
the inequality
\begin{equation}\label{eq3.26}
(\chi^{\star}_M(1)-\chi^{\star}(\tau))\Id_{\Omega_\kappa}=\xi_{M,\tilde w_\tau,\tilde w_1}\Id_{\Omega_\kappa}+\zeta\Id_{\Omega_\kappa},
\end{equation}
for some random variables $\zeta$ with $|\zeta|\leq \e_0$. Thus
\begin{equation}\label{eq3.27}
\E[(\chi^{\star}_M(1)-\chi^{\star}(\tau))^2]=
\E[(\chi^{\star}_M(1)-\chi^{\star}(\tau))^2\Id_{\Omega_\kappa}] + \E[(\chi^{\star}_M(1)-\chi^{\star}(\tau))^2\Id_{\Omega^c_\kappa}].
\end{equation}
Using (\ref{eq3.26}) we get
\begin{equation}
\E[(\chi^{\star}_M(1)-\chi^{\star}(\tau))^2\Id_{\Omega_\kappa}] = \E(\xi_{M,\tilde w_\tau,\tilde w_1}^2)-\E(\xi_{M,\tilde w_\tau,\tilde w_1}^2\Id_{\Omega^c_\kappa})+2\E(\zeta \xi_{M,\tilde w_\tau,\tilde w_1}\Id_{\Omega_\kappa})+\E(\zeta^2\Id_{\Omega_\kappa}).
\end{equation}
Using Cauchy-Schwarz and the fact that $|\zeta|\leq \e_0$, we get the bounds
\begin{equation}
\begin{aligned}
\E(\xi_{M,\tilde w_\tau,\tilde w_1}^2\Id_{\Omega^c_\kappa})&\leq \sqrt{\E(\xi_{M,\tilde w_\tau,\tilde w_1}^4)\Pb(\Omega^c_\kappa)}\leq C_1 (1-\tau)^{2/3}e^{-c\tilde \kappa^2/2},\\
|\E(\zeta \xi_{M,\tilde w_\tau,\tilde w_1}\Id_{\Omega_\kappa})|&\leq \e_0 \sqrt{\E(\xi_{M,\tilde w_\tau,\tilde w_1}^2)}\leq C_2 (1-\tau)^{1/3} \e_0,\\
\E(\zeta^2\Id_{\Omega_\kappa})&\leq \e_0^2,
\end{aligned}
\end{equation}
for some constants $C_1,C_2$ (since, as already mentioned, the $m$th moment of $\xi_{M,\tilde w_\tau,\tilde w_1}$ is of order $(1-\tau)^{m/3}$).

It remains to bound the last term of (\ref{eq3.27}). Let $\Lambda=\left\{|\chi^{\star}_M(1)-\chi^{\star}(\tau)|\leq \lambda\right\}$ and decompose $(\chi^{\star}_M(1)-\chi^{\star}(\tau))\Id_{\Omega^c_\kappa}$ as $(\chi^{\star}_M(1)-\chi^{\star}(\tau))\Id_{\Omega^c_\kappa}(\Id_{\Lambda}+\Id_{\Lambda^c})$. Then,
\begin{equation}
\E[(\chi^{\star}_M(1)-\chi^{\star}(\tau))^2\Id_{\Omega^c_\kappa}] \leq \E[(\chi^{\star}_M(1)-\chi^{\star}(\tau))^2\Id_{\Lambda^c}] +\lambda^2 \Pb(\Omega^c_\kappa).
\end{equation}
Integration by parts gives
\begin{equation}\label{eq74}
 \begin{aligned}
  &\E\left[(\chi^{\star}_M(1)-\chi^{\star}(\tau))^2\Id_{\Lambda^c}\right] = \lambda^2\Pb(|\chi^{\star}_M(1)-\chi^{\star}(\tau)|>\lambda)\\
  & +2\int_{\lambda}^{\infty}s\Pb(\chi^{\star}_M(1)-\chi^{\star}(\tau)>s)ds  -2\int_{-\infty}^{-\lambda}s\Pb(\chi^{\star}_M(1)-\chi^{\star}(\tau)\leq s)ds.
 \end{aligned}
\end{equation}
Now, for $s>0$,
\begin{equation}
\Pb(\chi^{\star}_M(1)-\chi^{\star}(\tau)>s) \leq \Pb(\chi^{\star}_M(1)\geq s/2)+\Pb(\chi^{\star}(\tau)\leq -s/2),
\end{equation}
and for $s<0$,
\begin{equation}
\Pb(\chi^{\star}_M(1)-\chi^{\star}(\tau)\leq s) \leq \Pb(\chi^{\star}_M(1)\leq s/2)+\Pb(\chi^{\star}(\tau)\geq -s/2).
\end{equation}
Recall that
\begin{equation}
\begin{aligned}
\Pb(\chi^{\star}_M(1)>s)&\leq\Pb(\chi^{\star}(1)>s),\\
\Pb(\chi^{\star}_M(1))\leq s)&\leq \Pb((1-\tau)^{1/3}\tilde{\cal A}_2(0)\leq s)=F_{\rm GUE}(s/(1-\tau)^{1/3}).
\end{aligned}
\end{equation}
Since both tails of $\chi^{\star}(1)$ and of the GUE Tracy-Widom distributions have (at least) exponential decay (see Appendix~\ref{sec:boundrandomIC} and~\ref{sec:boundstepStat}), it then follows that
\begin{equation}\label{eq3.33}
\E\left[(\chi^{\star}_M(1)-\chi^{\star}(\tau))^2\Id_{\Lambda^c}\right]\leq C \lambda^2 e^{-c\lambda/(1-\tau)^{1/3}}
\end{equation}
for some constants $C,c$.

Summing up we have obtained
\begin{multline}\label{eq3.37}
\E[(\chi^{\star}_M(1)-\chi^{\star}(\tau))^2] - \E(\xi_{M,\tilde w_\tau,\tilde w_1}^2)\\
= \Or\Big((1-\tau)^{2/3} e^{-c\tilde \kappa^2/2}; (1-\tau)^{1/3}\e_0; \e_0^2;\lambda^2e^{-c\tilde \kappa^2}; \lambda^2 e^{-c\lambda/(1-\tau)^{1/3}}\Big),
\end{multline}
with $\e_0=\Or(\kappa M (1-\tau)^{2/3})$.
Now we choose $M,\kappa,\lambda$. Let $\delta\in (0,1/3)$ be any fixed number and choose
\begin{equation}
M=\frac{1}{(1-\tau)^{\delta/2}},\quad \kappa=\frac{1}{(1-\tau)^{\delta/2}},\quad \lambda=1.
\end{equation}
Then, the error term in (\ref{eq3.37}) is just of order $\Or((1-\tau)^{1-\delta})$.
\end{proof}

Now we are ready to prove Theorem~\ref{thm:covGeneral}.
\begin{proof}[Proof of Theorem~\ref{thm:covGeneral}]
We have
\begin{equation}
\begin{aligned}
 \E\left[(\chi^{\star}(1)-\chi^{\star}(\tau))^2\right]& = \E\left[(\chi^{\star}(1)-\chi^{\star}_M(1))^2\right]  + \E\left[(\chi^{\star}_M(1)-\chi^{\star}(\tau))^2\right] \\
 &+ 2\E\left[(\chi^{\star}(1)-\chi^{\star}_M(1))(\chi^{\star}_M(1)-\chi^{\star}(\tau))\right]
\end{aligned}
\end{equation}
With the choice $M=(1-\tau)^{-\delta/2}$, by Lemma~\ref{lemma4} we have
\begin{equation}
 \E\left[(\chi^{\star}_M(1)-\chi^{\star}(\tau))^2\right] = \E[\xi_{M,\tilde w_\tau,\tilde w_1}^2]+\Or(1-\tau)^{1-\delta}.
\end{equation}
By Lemma~\ref{lemma2}, the Cauchy-Schwarz inequality, and $\xi_{\infty,\tilde w}=(1-\tau)^{1/3}\xi_{{\rm stat},\tilde w}$, we obtain
\begin{equation}
\E\left[(\chi^{\star}(1)-\chi^{\star}(\tau))^2\right] = (1-\tau)^{2/3}\E[\xi_{{\rm stat},\tilde w}]^2+\Or((1-\tau)^{1-\delta}).
\end{equation}
Since $\E[\xi_{{\rm stat},\tilde w}]=0$ the claimed result is proven.
\end{proof}

\section{Behavior around $\tau=0$ for droplet initial conditions}\label{sec:SeparateTimeBehavior}

Let us finally explain the asymptotic for $\tau\to 0$. The details are simple modifications of what we made for the case $\tau\to 1$. By Theorem~1 of~\cite{Pi17}, we have local weak convergence of the Airy$_2$ process to a Brownian motion for $\tau\to 0$,
\begin{equation}
\lim_{\tau\to0}\left(\tfrac{\tau}{1-\tau}\right)^{-1/3}\left(\tilde {\cal A}_2\left(\left(\tfrac{\tau}{1-\tau}\right)^{2/3} v\right) - \tilde {\cal A}_2(0)\right)=\sqrt{2}B(v).
\end{equation}
Lemma~\ref{lemma4} and Lemma~\ref{lemma5} can be easily readjusted for this case. Let us call $w_\tau=\tau^{2/3}\hat w_\tau$. Then, by Theorem~\ref{Thm1}, renaming $u=z \tau^{2/3}$,
\begin{equation}
\begin{aligned}
&\Cov\left(\chi^{\bullet}(\tau),\chi^{\bullet}(1)\right)\\
=&\Cov\left(\tau^{1/3}[\tilde{\cal A}_2(\hat w_\tau)-\hat w_\tau^2],\tau^{1/3}\max_{u\in\R}\left\{\tilde {\cal A}_2(z)-z^2+\left(\tfrac{1-\tau}{\tau}\right)^{1/3} {\cal A}_2(z \tfrac{\tau^{2/3}}{(1-\tau)^{2/3}})-z^2\tfrac{\tau}{1-\tau}\right\}\right)\\
=&\tau^{2/3}\Cov\left(\tilde {\cal A}_2(\hat w_\tau),\max_{z\in\R}\left\{\tilde{\cal A}_2(z)-z^2+\sqrt{2}B(z)\right\}+\left(\tfrac{1-\tau}{\tau}\right)^{1/3}{{\cal A}}_2(0)\right)+\Or(\tau^{1-\delta})\\
=&\tau^{2/3}\Cov\left(\tilde {\cal A}_2(\hat w_\tau),\max_{z\in\R}\left\{\tilde{\cal A}_2(z)-z^2+\sqrt{2}B(z)\right\}\right)+\Or(\tau^{1-\delta}),
\end{aligned}
\end{equation}
for any $\delta>0$, where the covariance of $\tilde {\cal A}_2(\hat w_\tau)$ and ${{\cal A}}_2(0)$ is zero, since they are independent processes. The second term in the covariance has the same distribution as $\xi_{{\rm stat},0}$, which is has expected value $0$. This leads to the claimed result of Theorem~\ref{thm:covstep0}.

\appendix
\section{Bounds on point-to-point LPP}\label{sec:boundstepLPP}
In the proofs, we use known results for the point-to-point LPP with exponential random variables, which we recall here.
\begin{prop}\label{PropBoundsPP}
For $\eta\in(0,\infty)$ define $\mu=(\sqrt{\eta \ell}+\sqrt{\ell})^2$, $\sigma=\eta^{-1/6}(1+\sqrt{\eta})^{4/3}$, and the rescaled random variable
\begin{equation}
L^{\rm resc}_\ell:=\frac{L_{(0,0)\to(\eta\ell,\ell)}-\mu}{\sigma\ell^{1/3}}.
\end{equation}
(a) Limit law
\begin{equation}
\lim_{\ell\to\infty} \Pb(L^{\rm resc}_\ell\leq s) = F_{\rm GUE}(s),
\end{equation}
with $F_{\rm GUE}$ the GUE Tracy-Widom distribution function.\\
(b) Bound on upper tail: there exist constants $s_0,\ell_0,C,c$ such that
\begin{equation}
\Pb(L^{\rm resc}_\ell\geq s)\leq C e^{-c s}
\end{equation}
for all $\ell\geq \ell_0$ and $s\geq s_0$.\\
(c) Bound on lower tail: there exist constants $s_0,\ell_0,C,c$ such that
\begin{equation}\label{eqA4}
\Pb(L^{\rm resc}_\ell\leq s)\leq C e^{-c |s|^{3/2}}
\end{equation}
for all $\ell\geq \ell_0$ and $s\leq -s_0$.
\end{prop}
\emph{(a)} was proven in Theorem~1.6 of\cite{Jo00b}. Using the relation with the Laguerre ensemble of random matrices (Proposition~6.1 of~\cite{BBP06}), or to TASEP described above, the distribution is given by a Fredholm determinant. An exponential decay of its kernel leads directly to \emph{(b)}. See e.g.\ Proposition~4.2 of~\cite{FN13} or Lemma~1 of~\cite{BFP12} for an explicit statement. \emph{(c)} was proven in~\cite{BFP12} (Proposition~3 together with (56)). In the present language it is reported in Proposition~4.3 of~\cite{FN13} as well.

\section{Bounds for point-to-line LPP}\label{sec:largedev_flat}

\begin{prop}\label{propBoundflat}
Let ${\cal L}=\{(k,-k),k\in\Z\}$. Consider the rescaled LPP from $\cal L$ to $(\ell,\ell)$ given by
\begin{equation}
L^{{\cal L},\rm resc}_\ell=\frac{L_{{\cal L}\to (\ell,\ell)}-4\ell}{2^{4/3}\ell^{1/3}}.
\end{equation}
(a) Limit law
\begin{equation}
\lim_{\ell\to\infty} \Pb(L^{{\cal L},\rm resc}_\ell\leq s)=F_{\rm GOE}(2^{2/3} s).
\end{equation}
(b) Bound on upper tail: there exists constants $s_0,\ell_0,C,c$ such that
\begin{equation}\label{eqD2}
\Pb(L^{{\cal L},\rm resc}_\ell\geq s)\leq C e^{-c s}
\end{equation}
for all $\ell\geq \ell_0$ and $s\geq s_0$.\\
(c) Bound on lower tail: there exists constants $s_0,\ell_0,C,c$ such that
\begin{equation}
\Pb(L^{{\cal L},\rm resc}_\ell\leq s)\leq C e^{-c |s|^{3/2}}
\end{equation}
for all $\ell\geq \ell_0$ and $s\leq -s_0$.
\end{prop}
(a) was obtained in~\cite{Sas05,BFPS06} in terms of TASEP, which can be directly rewritten in term of LPP (the complete proof is present in~\cite{BF07}). For general slopes of $\cal L$ it was shown in~\cite{FO17}. (b) this tails follows from the asymptotic analysis on the correlation kernel made in~\cite{BF07}. (c) It follows from (\ref{eqA4}) since $\Pb(L_{{\cal L}\to (\ell,\ell)}\leq x)\leq \Pb(L_{(0,0)\to(\ell,\ell)}\leq x)$.

\section{Bounds on LPP with random initial condition}\label{sec:boundrandomIC}
\begin{prop}\label{propRandomBounds}
Define $L^\sigma_{{\cal L}\to(\ell,\ell)}=\max_{k} \{L_{(k,-k)\to (\ell,\ell)}+h^0(k,-k)\}$ with $h^0$ as in (\ref{eq1.7}), and consider the rescaled LPP time
\begin{equation}
L^{\sigma,\rm resc}_\ell=\frac{L^\sigma_{{\cal L}\to(\ell,\ell)}-4\ell}{2^{4/3}\ell^{1/3}}.
\end{equation}
Then, there exists constants $s_0,\ell_0,C,c$ such that:\\
(a) Bound on upper tail:
\begin{equation}
\Pb(L^{\sigma,\rm resc}_\ell\geq s)\leq C e^{-c s}
\end{equation}
for all $\ell\geq \ell_0$ and $s\geq s_0$. \\
(b) Tail on lower tail:
\begin{equation}
\Pb(L^{\sigma,\rm resc}_\ell\leq s)\leq C e^{-c |s|^{3/2}}
\end{equation}
for all $\ell\geq \ell_0$ and $s\leq -s_0$.
\end{prop}
\begin{proof}
(a) Define $J(u)=u(2\ell)^{2/3}(1,-1)$ and $W_\ell(u)=h^0(J(u))/(2^{4/3}\ell^{1/3})$. By Donsker's theorem, $u\mapsto W_\ell(u)$ converges weakly to a two-sided Brownian motion with diffusion coefficient $2\sigma^2$.
Further, define
\begin{equation}
L^{\rm pp}_\ell(u):=\frac{L_{J(u)\to (\ell,\ell)}-4\ell}{2^{4/3}\ell^{1/3}}.
\end{equation}
Then, we can write
\begin{equation}
L^{\sigma,\rm resc}_\ell = \max_u \{L^{\rm pp}_\ell(u)+W_\ell(u)\} \leq \max_u \{L^{\rm pp}_\ell(u)+u^2/2\}+\max_{u} \{W_\ell(u)-u^2/2\}.
\end{equation}
Thus,
\begin{equation}
\Pb(L^{\sigma,\rm resc}_\ell\geq s)\leq \Pb(\max_u \{L^{\rm pp}_\ell(u)+u^2/2\}\geq s/2)+\Pb(\max_{u} \{W_\ell(u)-u^2/2\}\geq s/2).
\end{equation}
By computations based on Doob maximal inequality (used for instance in (\ref{eq3.31})), one obtains $\Pb(\max_{u} \{W_\ell(u)-u^2/2\}\geq s/2)\leq C e^{-c s^2}$ for some constants $C,c>0$. To bound the first term without new estimates, remark that for any $M$ we can bound
\begin{equation}
\begin{aligned}
\Pb(\max_u \{L^{\rm pp}_\ell(u)+u^2/2\}\geq s/2) &\leq \Pb(\max_{u} L^{\rm pp}_\ell(u)\geq s/4-M^2/2)\\ &+ \Pb(\max_{|u|>M} \{L^{\rm pp}_\ell(u)+u^2/2\}\geq s/4)
\end{aligned}
\end{equation}
The exponential decay in $s$ for the second term is just a special case of (\ref{eq24a}) (set $\tau=0$) and it holds for all $M\geq M_0$, for some finite $M_0$. We fix $M=M_0$ and then, using the fact that $\max_u L^{\rm pp}_\ell(u)=L^{{\cal L},\rm resc}_\ell$, by (\ref{eqD2}) we have exponential decay in $s$ for the first term as well.

(b) It follows from (\ref{eqA4}) since $\Pb(L^\sigma_{{\cal L}\to(\ell,\ell)}\leq x)\leq \Pb(L_{(0,0)\to(\ell,\ell)}\leq x)$.
\end{proof}

\section{Bounds on stationary LPP}\label{sec:boundstepStat}
We now state and give a short proof of the tails of the one-point distribution in the stationary case with $\rho=1/2$ of the LPP to $(\ell,\ell)$.
\begin{prop}\label{propStatBounds} Let $\rho=1/2$. Then there exists constants $s_0,\ell_0,C,c$ such that:\\
(a) Bound on upper tail:
\begin{equation}\label{eqB2}
\Pb(\LB_{(0,0)\to(\ell,\ell)}\geq 4\ell+2^{4/3} s \ell^{1/3})\leq C e^{-c s}
\end{equation}
for all $\ell\geq \ell_0$ and $s\geq s_0$. \\
(b) Bound on lower tail:
\begin{equation}
\Pb(\LB_{(0,0)\to(\ell,\ell)}\leq 4\ell+2^{4/3} s \ell^{1/3})\leq C e^{-c |s|^{3/2}}
\end{equation}
for all $\ell\geq \ell_0$ and $s\leq -s_0$.
\end{prop}
\begin{proof}
(a) One can write $\LB_{(0,0)\to(\ell,\ell)}=\max\{L^{\vert,\rho}(\ell,\ell),L^{-,\rho}(\ell,\ell)\}$, where $L^{\vert,\rho}(\ell,\ell)$ (resp.\ $L^{-,\rho}(\ell,\ell)$) are the LPP with one-sided perturbation only on $i=0$ (resp. $j=0$). Then,
\begin{equation}
\Pb(\LB_{(0,0)\to(\ell,\ell)}\geq x)\leq \Pb(L^{\vert,\rho}(\ell,\ell)\geq x)+\Pb(L^{-,\rho}(\ell,\ell)\geq x).
\end{equation}
By choosing $x=4\ell+s 2^{4/3} \ell^{1/3}$, Lemma~3.3 of~\cite{FO17} (based on the estimates on the correlation kernel in~\cite{BBP06}) gives exponential decay in $s$ for all $s\geq s_0$.

(b) It follows from (\ref{eqA4}), since $\Pb(\LB_{(0,0)\to(\ell,\ell)}\leq x)\leq \Pb(L_{(0,0)\to(\ell,\ell)}\leq x)$.
\end{proof}

\begin{lem}\label{lemmaIncrementsStat}
Let $\rho=1/2$ and define $I(u)=(\ell-2u\ell^{2/3},\ell+2u\ell^{2/3})$. Then, for any $\alpha>0$, we have
\begin{equation}\label{eqB4}
\Pb(|\LB_{(0,0)\to I(K)} -\LB_{(0,0)\to I(0)}|\geq \alpha \ell^{1/3})\leq 4 e^{-\alpha^2/(16 K)}
\end{equation}
for all $\ell$ large enough. Furthermore,
\begin{equation}\label{eqB5}
\Pb(\max_{u\geq K}\LB_{(0,0)\to I(u)} -\LB_{(0,0)\to I(K)}-\beta u^2 \ell^{1/3}\geq \alpha \ell^{1/3})\leq  C e^{-\frac{(\alpha+\beta K^2)^2}{16 K}},
\end{equation}
for a constant $C$ and for all $\beta>0$ and $\alpha>-\beta K^2$ and $\ell$ large enough.
\end{lem}
\begin{proof}
The process $K\mapsto Y_K:=\LB_{(0,0)\to I(K)} -\LB_{(0,0)\to I(0)}$ is a martingale~\cite{BCS06} given by a sum of \iid zero mean random variables $Z_j-2$, with $Z_j\sim {\rm Exp}(1/2)$. By the exponential Chebyshev inequality,
\begin{equation}
\begin{aligned}
\Pb(|Y_K|\geq \alpha\ell^{1/3})& \leq \Pb(Y_K\geq \alpha \ell^{1/3}) + \Pb(-Y_K\geq \alpha\ell^{1/3})\\
&\leq \inf_{t\geq 0} e^{-t\alpha \ell^{1/3}} \E(e^{t (Z_1-2)})^{2K\ell^{2/3}} + \inf_{t'\geq 0} e^{-t'\alpha\ell^{1/3}} \E(e^{-t'(Z_1-2)})^{2K\ell^{2/3}}.
\end{aligned}
\end{equation}
Using $\E(e^{t (Z_1-2)})=\frac{e^{-2t}}{1-2t}$ for $t\in (0,1/2)$ and $\E(e^{-t'(Z_1-2)})=\frac{e^{2 t'}}{1+2t'}$ for all $t'\geq 0$, after the minimization we obtain
\begin{equation}
\Pb(|Y_K|\geq \alpha\ell^{1/3}) \leq 2 e^{-\alpha^2/(16 K)(1+\Or(\alpha K^{-1} \ell^{-1/3})}\leq 4 e^{-\alpha^2/(16 K)}
\end{equation}
for all $\ell$ large enough.

For the second estimate, from the inequality
\begin{equation}\label{eqB8}
\begin{aligned}
&\Pb(\max_{u\geq K}\LB_{(0,0)\to I(u)} -\LB_{(0,0)\to I(K)}-\beta u^2 \ell^{1/3}\geq \alpha \ell^{1/3}) \\
\leq &\sum_{m\geq 1}\Pb(\max_{u\in [K m,K(m+1)]}\LB_{(0,0)\to I(u)} -\LB_{(0,0)\to I(K)}\geq (\alpha+\beta K^2 m^2) \ell^{1/3})\\
\leq & \sum_{m\geq 1} \inf_{t>0} e^{-t(\alpha+\beta K^2 m^2)\ell^{1/3}} \E(e^{t(Z_1-2)})^{2K m \ell^{2/3}}.
\end{aligned}
\end{equation}
Maximising over $t$ and taking the sum we finally get\footnote{To be precise, for $\e>0$ small, one can bound $\Pb(\LB_{(0,0)\to I(u)} -\LB_{(0,0)\to I(K)}\geq (\alpha+\beta K^2 u^2) \ell^{1/3})$ for all $u\geq \e K \ell^{1/3}$ using (\ref{eqB8}) and for $m\in\{1,\ldots,\e \ell^{1/3}\}$ we can minimize over $t$ and compute the series expansion in the exponent for large $\ell$.}
\begin{equation}
\Pb\Big(\max_{u\geq K}\LB_{(0,0)\to I(u)} -\LB_{(0,0)\to I(-K)}-\beta u^2 \ell^{1/3}\geq \alpha \ell^{1/3}\Big) \leq C e^{-\frac{(\alpha+\beta K^2)^2}{16 K}}
\end{equation}
for a constant $C$ and for all $\beta>0$ and $\alpha>-\beta K^2$ and $\ell$ large enough.
\end{proof}

\section{Bounds for point-to-half line LPP}\label{sec:boundshalfline}
\begin{prop}\label{PropBoundhalfflat}
Let $I(u)=(\tau N,\tau N)+u(2 N)^{2/3}(1,-1)$. Then,
\begin{equation}
\Pb\Big(\max_{|u|>M} L_{I(u)\to (N,N)}>4(1-\tau)N+2^{4/3}(s-\gamma M^2) N^{1/3}\Big)\leq C  e^{-c M^2 (1-\tau)^{-4/3}} e^{-\tilde c s (1-\tau)^{-1/3}}
\end{equation}
for some constants $C,c,\tilde c>0$, which can be taken uniform in $N$ and uniform for $\gamma$ in a compact subset of $(0,1/(1-\tau))$.
\end{prop}
\begin{proof}
By symmetry, it is enough to get the bound on the distribution of $\max_{u<-M} L_{I(u)\to (N,N)}$. By first shifting $I(-M)$ to the origin, and then using the mapping between LPP and TASEP, the distribution function is the same as the distribution of TASEP particle number $n=t/4+\tilde \tau (t/2)^{2/3}$ at time $t=4(1-\tau)N+2^{4/3} N^{1/3}(s-\gamma M^2)$, starting at $x_k(0)=-2k$, $k\geq 0$.

From Proposition~3 of~\cite{BFS07} we have an explicit expression in terms of Fredholm determinant. The upper tail estimate is standard. Using Hadamard's bound it is enough to have a bound on the correlation kernel. In Section~4 of~\cite{BFS07} exponential decay of the rescaled correlation kernel has been proven. Then, simple algebraic computations give the claimed result.
\end{proof}


\end{document}